\documentclass[11pt]{article}

\usepackage[margin=1in]{geometry}
\usepackage{amsmath,amssymb,amsthm}
\usepackage{mathtools}
\usepackage{algorithm}
\usepackage[noend]{algpseudocode}
\usepackage[textsize=tiny,textwidth=2cm]{todonotes}

\usepackage{hyperref}
\hypersetup{hidelinks}

\newtheorem{theorem}{Theorem}[section]
\newtheorem{lemma}[theorem]{Lemma}
\newtheorem{proposition}[theorem]{Proposition}
\theoremstyle{remark}

\newcommand{\R}{\mathbb{R}}
\newcommand{\C}{\mathbb{C}}
\newcommand{\E}{\mathbb{E}}
\newcommand{\Pp}{\mathbb{P}}

\newcommand{\poly}{\mathrm{poly}}

\newcommand{\cD}{{\mathcal D}}
\newcommand{\cM}{{\mathcal M}}

\newcommand{\cR}{{\mathcal R}}
\newcommand{\cB}{{\mathcal B}}
\newcommand{\cI}{{\mathcal I}}
\newcommand{\bx}{{\mathbf x}}
\newcommand{\bt}{{\mathbf t}}
\newcommand{\bv}{{\mathbf v}}

\newcommand{\be}{{\mathbf e}}
\newcommand{\bj}{{\mathbf j}}

\newcommand{\by}{{\mathbf y}}
\newcommand{\bz}{{\mathbf z}}
\newcommand{\bu}{{\mathbf u}}
\newcommand{\bw}{{\mathbf w}}

\newcommand{\bzero}{\mathbf{0}}
\newcommand{\rank}{\mathrm{rank}}
\newcommand{\Vol}{\mathrm{Vol}}

\newcommand{\OPT}{\mathrm{OPT}}

\newcommand{\ksinline}[1]{%
\todo[inline,color=yellow!20]{KS: #1}%
}

\title{Determinant maximization subject to a partition matroid constraint
via stable distributions}

\author{Yihang Sun \\ Stanford University \and Jan Vondr\'ak \\ Stanford University}

\begin{document}

\maketitle

\begin{abstract}
Given vectors $\mathbf{v}_i \in \mathbb{R}^d$, we consider the problem of choosing a set $I$ independent in a partition matroid in order to maximize the determinant $\det (\sum_{i \in I} \mathbf{v}_i \mathbf{v}_i^T)$. Our main result is a polynomial-time approximation algorithm that finds a solution of value $\det ( \sum_{i \in I} \mathbf{v}_{i} \mathbf{v}_{i}^T) \geq e^{-O(d)} \mathrm{OPT}$, where $\mathrm{OPT} = \max_{I^* \in \mathcal{I}} \det ( \sum_{i \in I^*} \mathbf{v}_{i} \mathbf{v}_{i}^T)$. For partition matroids of rank $m \leq d$, we give a similar result for approximating the $m$-dimensional volume spanned by the chosen vectors, within a factor of $e^{O(m)}$.

This matches earlier known algorithms that estimate the optimal value but do not find the corresponding solution, up to a constant in the exponent. Similar to these estimation algorithms,
our algorithm is based on the saddle-point relaxation proposed by Nikolov and Singh \cite{NS16}. A new ingredient is a randomized transformation based on $1/2$-stable distributions, which converts the saddle-point relaxation into a more convenient multilinear relaxation.
\end{abstract}

\newpage 
\tableofcontents
\newpage 
\section{Introduction}

We consider the problem of determinant maximization: We are given finitely many vectors $\bv_{i} \in \R^d$. Our goal is to choose a set of indices $I$ subject to a certain constraint, so as to maximize the determinant $\det (\sum_{i \in I} \bv_{i} \bv_{i}^T)$. Note that in the case where $|I|=d$, this is the square of the volume spanned by the chosen vectors. In the case of $|I|<d$, the determinant as defined above would be identically $0$; a natural modification of the objective in this case is (the square of) the $m$-dimensional volume spanned by the selected vectors.

This problem is motivated by sparsification of datasets: Suppose we have data items with features represented by the vectors $\bv_{i}$, and we want to choose a sparse subset which retains the diversity of the original dataset as much as possible. This objective arises naturally in machine learning, information retrieval, statistics and optimization. Diversity can be measured in different ways. In this paper, we focus on the determinant, which has emerged as a popular measure of diversity for the following reasons: It makes sure that the selected subset does not collapse to a lower-dimensional subspace, and further that different dimensions of the dataset are covered as much as possible. Besides the diversity objective, one can also consider various constraints on the subset to be selected. The most basic constraint is a cardinality constraint: choosing at most $m$ vectors overall. The constraint considered here is the next natural choice: We have certain types of data items, and we want to choose $1$ of each type (or more generally a given number of each type). This can be described as a {\em partition matroid constraint}, and this is the constraint studied in this paper. Beyond this constraint, one may study more general matroid constraints (which are natural due to their structural properties), as well as other types of constraints arising from various applications. We review the state of the art for different variants of determinant maximization more thoroughly in Section~\ref{sec:prior}. 

In the case of a cardinality constraint, a seminal result was shown by Nikolov \cite{Nikolov15}: Given  integers $1 \leq m \leq d$ and vectors in $\R^d$, there is an algorithm which selects $m$ vectors such that the $m$-dimensional volume spanned by them approximates the optimal solution within a factor of $e^{m/2+o(m)}$. His algorithm is based on a convex relaxation, related to approximations of convex bodies by ellipsoids. While the exponential factor might seem large, it is a meaningful positive result in this context: our solution is losing only a constant factor ``per dimension''. It was in fact proved that the problem under a cardinality constraint in the $m=d$ case is NP-hard to approximate within a factor better than $c^d$ for some $c>1$ \cite{Koutis06,CivrilM13,SEFM15}. 

Beyond the cardinality constraint, the case of partition matroid constraints was considered in \cite{NS16}. It turns out that this seemingly small change in the constraint makes the problem substantially more difficult. A natural extension of the relaxation from \cite{Nikolov15} does not work for partition constraints, and a much more intricate relaxation which was found by \cite{NS16} is required here. This ``saddle-point relaxation'' has both maximization and minimization variables and a concave/convex structure. Its design was inspired by the theory of stable polynomials, and in particular Gurvits' proof of the Van der Waerden conjecture for the permanent of doubly stochastic matrices \cite{Gurvits08}. Nikolov and Singh found an elegant way to relate the optimal value of this relaxation to the discrete optimum, and proved that the ratio between the two quantities is at most $e^m$ (where $m$ is the number of vectors to select). Hence, their work provides an $e^{m}$-{\em estimation} algorithm for the maximum determinant value. Nevertheless,
a remaining stumbling block was how to find a feasible solution, given a fractional solution provided by the saddle-point relaxation.

The discovery of this relaxation by \cite{NS16} has been hugely influential in further studies of the determinant maximization problem, and related problems. Straszak and Vishnoi \cite{SV17} extended this relaxation and the $e^{O(m)}$ bound on its integrality gap to other classes of matroids, in particular {\em strongly Rayleigh matroids} which are related to stable polynomials. 
Anari et al. \cite{AGSS17,AMGV18} used similar relaxations for developing approximation algorithms for variants the Nash social welfare problem. These developments were followed by an exciting sequence of works on connections between sampling, optimization, and stable / completely log-concave polynomials \cite{AG17,AGGS17,AGV18,ALGV19}, which shares some technical features with the determinant maximization problem.

The first approximation algorithm finding a feasible solution of guaranteed value for determinant maximization under a partition matroid constraint (and more generally a regular matroid constraint) was proposed by Ebrahimi, Straszak and Vishnoi \cite{ESV17}. Their approximation guarantee depends on the sizes of $|J_i|$; using our notation, for $|J_i| = \poly(m)$, it gives an $m^{O(m)}$-approximation algorithm. This algorithm does not use the saddle-point relaxation of \cite{NS16}. It uses a simpler relaxation which also plays a role in our paper; we discuss this in Section~\ref{sec:prior}.

Coming back to the saddle-point relaxation, Madan, Nikolov, Singh and Tantipongpipat \cite{MNST20} extended the framework to general matroids and gave an approximation algorithm that finds a feasible solution within a factor of $e^{O(d^3)}$ in the case of partition matroids, and $d^{O(d^3)}$ in the case of general matroids. These are the first approximation algorithms whose performance depends only on the dimension $d$. An important ingredient they developed towards this end is a {\em sparsification procedure} for fractional solutions of the saddle-point relaxation, proving that the number of fractional variables can be reduced to $\poly(d)$.

A different elegant approach to determinant maximization, based on the structure of matroid intersection, was developed in \cite{BLPST22}. This approach gives a $m^{O(m)}$-approximation algorithm for $m \leq d$ and any general matroid constraint, and a follow-up paper gives a $d^{O(d)}$-approximation for $d\le m$. 

In spite of these advances, it was not known whether the optimal approximation guarantee of $e^{O(d)}$ can be achieved for determinant maximization subject to a simple partition matroid constraint, even in the case of $m=d$. 

\subsection{Our results}

We design the first polynomial-time algorithm for determinant maximization subject to a partition matroid constraint with an optimal approximation factor up to a constant in the exponent.

\begin{theorem}
\label{thm:main}
There is a randomized polynomial-time algorithm which for any partition matroid $\cM = (E, \cI)$, $\rank(\cM) = m \geq d \geq 1$, and vectors $\bv_i \in \R^d$ for $i \in E$, finds an independent set $I$ of $\cM$ such that with high probability
$$ \det \left( \sum_{i \in I} \bv_{i} \bv_{i}^T \right) \geq e^{-O(d)} \ \OPT $$
where $\OPT = \max_{I^* \in \cI} \det \left( \sum_{i \in I^*} \bv_{i} \bv_{i}^T \right)$.
\end{theorem}

We remark that for simple partition matroids, the bound is $e^{-6d}$ for $m=d$ and $e^{-9d}$ for $m>d$, and the algorithm can be implemented as a Las Vegas algorithm (terminating with a guaranteed solution in expected polynomial time). For general partition matroids, the algorithm is Monte Carlo and the bound is approximately $e^{-14d}$ (but we did not attempt to optimize the exponents).
For $m < d$, we have a similar result, with an objective appropriate for this setting.

\begin{theorem}
\label{thm:main2}
There is a randomized (expected) polynomial-time algorithm which for any partition matroid $\cM = (E, \cI)$, $\rank(\cM) = m \leq d$, and vectors $\bv_i \in \R^d$ for $i \in E$, finds a basis of $\cM$ such that 
$$ \Vol_m [ \bv_i: i \in I] \geq e^{-O(m)} \ \OPT $$
where $\Vol_m[\bv_i:i  \in I]$ denotes the $m$-dimensional volume of the parallelepiped formed by the respective vectors, and $\OPT = \max_{I^* \in \cI} \Vol_m [ \bv_i: i \in I^*]$.
\end{theorem}

\subsection{Our techniques}

For most of this paper, we restrict our attention to the case of simple partition matroids, where we select one vector from each part. In this case it is natural to denote the vectors by $\bv_{ij}, j \in J_i$, and look for a choice of one vector $\bv_{i j_i}$ for each $i$. The more general result for partition matroids is obtained by a black box reduction.

Our starting point is the saddle-point relaxation of Singh and Nikolov \cite{NS16}. It was known that an optimal fractional solution $\bx$ can be rounded to a feasible discrete solution in principle, but it was not known how to do it in polynomial time. The interpretation of $\bx$ is that it defines a linear combination of rank $1$ matrices, $\sum_{i=1}^{m} \sum_{j \in J_i} x_{ij} \bv_{ij} \bv_{ij}^T$, whose determinant is related to our objective.
A crucial step in our algorithm is to convert the fractional solution $\bx$ to another fractional solution $\by$ with a different interpretation: $y_{ij}$ being the coefficients in a linear combination defining vectors $\bw_i = \sum_{j \in J_i} y_{ij} \bv_{ij}$, such that $\det (\sum_{i=1}^{m} \bw_i \bw_i^T)$ has a good value relative to the optimum.
If we can find such a fractional solution $\by$, then we are in a good shape since we can round $\by$ to an integer vector, using the multilinear properties of the determinant. The main hurdle is how to transform the fractional solution $\bx$ into $\by$.

The solution relies on the notion of {\em $1/2$-stable distributions} (unrelated to stable polynomials): These are distributions such that a linear combination of independent random variables $Z = \sum c_j Z_j$ where $\sum |c_j|^{1/2} = 1$ has again the same distribution. This distribution also has the curious property that $Z' = \sum c_j |Z_j|$ has a distribution not radically different from $|Z| = |\sum c_j Z_j|$, in terms of the logarithmic expectations $\E[\log Z']$ and $\E[\log |Z|]$ (both being relatively small constants). If we replace the fractional solution $x_{ij}$ by $y_{ij} = \frac{x_{ij}^2 Z_{ij}}{\sum_{j' \in J_i} x_{ij'}^2 |Z_{ij'}|}$ and define $\bw_i = \sum_{j \in J_i} y_{ij} \bv_{ij}$, it turns out that the value of the fractional solution $\det (\sum_{i=1}^{m} \bw_i \bw_i^T)$ can be related to the value of the saddle-point relaxation. The analysis, however, is quite intricate, using an inductive argument based on the Gibbs inequality for relative entropy. 

With this transformation in place, the algorithm is straightforward in the case of $m=d$. Some additional care is needed to handle the cases of $m>d$ and $m<d$. In the case of $m>d$, an additional characterization of the distribution of values over $d \times d$ minors is needed, which follows from the KKT conditions for a saddle-point optimum. This could be of independent interest.

We perform all the analysis for the case of simple partition matroids (with one vector to be selected from each part). For the case of general partition matroids, we use a black box reduction which does not significantly increase the approximation factor. The approximation factor could be probably improved using a direct analysis, but we do not pursue this in this paper.

\subsection{AI usage statement}

The authors tried various approaches to resolve this problem in the past, but a remaining hurdle was  how to define the transformation from $\sum_{j \in J_i} x_{ij} \bv_{ij} \bv_{ij}^T$ to $\bw_i = \sum_{j \in J_i} y_{ij} \bv_{ij}$ (or another manageable formulation) in a way that preserves the determinant up to a bounded factor. The solution using $1/2$-stable distributions for simple partition matroids and $m=d$ was found by GPT-5.6 Sol, in mid-August 2026. The authors then digested the solution and extended it to $m<d$ and $m>d$, as well as general partition matroids, also with the help of GPT-5.6 Sol. GPT-6 Astra was released during the editing of this paper but we did not use it. The paper is written entirely by hand, with a discussion of the intuition behind the solution as the authors understand it.

\section{Preliminaries}

Let us formally define the problem of determinant maximization subject to a partition matroid constraint. We first consider the case of a simple partition matroid (choosing 1 element per part).

\paragraph{Simple partition matroid.}
A simple partition matroid is defined by a partition of the ground set $N$ into $J_1 \cup J_2 \cup \ldots \cup J_m$. The family of independent sets is $\cI = \{ I \subset N: \forall i \in [m], |I \cap J_i| \leq 1 \}$. The bases are maximal independent sets, $\cB = \{ I \subset N: \forall i \in [m], |I \cap J_i| = 1 \}$.
The determinant maximization problem can be formulated in several equivalent ways; we state the most basic variant as follows. In the following, $[\bv_{1j_1}, \ldots, \bv_{dj_d} ]$ denotes a matrix whose columns are the respective vectors in $\R^d$.

\

\noindent {\bf Choosing $m = d$ vectors.}

{\bf Input:} Vectors $\bv_{i,j} \in \R^d$ for $1 \leq i \leq d$, $j \in J_i$ where $J_i$ is a nonempty finite set for each $i$.

{\bf Output:} $(j_1,\ldots,j_d), j_i \in J_i$, such that $\det^2 [\bv_{1,j_1}, \ldots, \bv_{d,j_d}]$ is (approximately) maximized.

\

More generally, we can consider instances where the number of vectors to choose does not match the dimension $d$. In that case we need to modify the objective function to be meaningful.

\

\noindent {\bf Choosing $m \geq d$ vectors.}

{\bf Input:} Vectors $\bv_{i,j} \in \R^d$ for $1 \leq i \leq m$, $j \in J_i$ where $J_i$ is a nonempty finite set for each $i$.

{\bf Output:} $(j_1,\ldots,j_m), j_i \in J_i$, such that $\det \left(\sum_{i=1}^{m} \bv_{i,j_i} \bv_{i,j_i}^T\right)$ is (approximately) maximized.

\

\noindent {\bf Choosing $m \leq d$ vectors.}

{\bf Input:} Vectors $\bv_{i,j} \in \R^d$ for $1 \leq i \leq m$, $j \in J_i$ where $J_i$ is a nonempty finite set for each $i$.

{\bf Output:} $(j_1,\ldots,j_m), j_i \in J_i$, such that $\Vol^2_m \left[ \bv_{1,j_1}, \ldots, \bv_{m,j_m} \right]$ is (approximately) maximized.

\

It can be checked that these objectives coincide in the case where $m=d$. We denote the value of the optimal solution by  $\OPT$. Note that in the case of $m=d$, this is a square of the determinant formed by the respective vectors.

\

\noindent {\bf General partition matroid.}
More generally, we consider a partition matroid constraint where multiple elements can be chosen from each part: $\cI = \{ I \subset N: \forall i \in [r], |I \cap J_i| \leq k_i \}$.
We define the determinant maximization problem in this setting as follows.

\

{\bf Input:} Vectors $\bv_{i,j} \in \R^d$ for $1 \leq i \leq r$, $j \in J_i$, and integers $k_i \geq 1$, $\sum_{i=1}^{r} k_i = m$.

{\bf Output:} $(K_1,\ldots,K_r)$ such that $K_i \subseteq J_i$ and $|K_i|=k_i$ for each $i$, and
\begin{itemize}
\item for $m \geq d$, $\det \left(\sum_{i=1}^{r} \sum_{j \in K_i} \bv_{i,j} \bv_{i,j}^T\right)$ is (approximately) maximized;
\item for $m \leq d$, $\Vol_m^2 \left[ \bv_{i,j}: 1 \leq i \leq r, j \in K_i \right]$ is (approximately) maximized.
\end{itemize}

Here, $\Vol_m$ denotes the $m$-dimensional volume of the parallelepiped spanned by the vectors.

\section{High-level overview}
\label{sec:overview}

Let us explain our approach in more detail, by positioning it in the context of known approaches to the determinant maximization problem. We focus on the basic case of simple partition matroids, with $m=d$.

\subsection{Prior techniques}
\label{sec:prior}

The basic vehicle for development of algorithms for this problem has been the ``saddle-point'' (or Nikolov-Singh) relaxation \cite{NS16}. In the $m=d$ case, it is formulated as follows.

\paragraph{The saddle-point relaxation.}
\begin{equation}
\label{eq:saddle-point}
  \begin{aligned}
\cR =  \max_\bx \inf_\bz & \det \left (\sum_{i=1}^{d} \sum_{j \in J_i} e^{z_i} x_{ij} \bv_{ij} \bv_{ij}^T \right) \\
\forall i \in [d]; & \sum_{j \in J_i} x_{ij} = 1, \\
\forall i \in[d], j \in J_i; & x_{ij} \geq 0, \\
    & \sum_{i=1}^{d} z_i \geq 0.
    \end{aligned}
\end{equation}

This relaxation can be solved efficiently thanks to log-concavity in $\bx$ and log-convexity in $\bz$. It is also relatively easy to show that it forms an upper bound on the discrete optimum,  $\OPT \leq \cR$.
The main result of \cite{NS16}, based on a lemma of Gurvits \cite{Gurvits08}, is that the optimum $\cR$ provides an estimate of the discrete optimum, in the following sense.

\begin{theorem}
\label{thm:NS-estimate}
For an optimal fractional solution $(\bx,\bz)$, we have
\begin{equation}
\label{eq:multilinear}
 D(\bx) := \sum_{(j_1,\ldots,j_d) \in J_1 \times \ldots J_d} \prod_{i=1}^{d} x_{i j_i} \det \left( \sum_{i=1}^{d} \bv_{i j_i} \bv_{i j_i}^T \right) 
\end{equation}
\begin{equation*}
\geq \frac{d!}{d^d} \det \left( \sum_{i=1}^{d} \sum_{j \in J_i} e^{z_i} x_{ij} \bv_{ij} \bv_{ij}^T \right) \geq \frac{1}{e^d} \cR.
\end{equation*}

I.e., a random solution sampled with probabilities $x_{i,j_i}$ has expected value at least $e^{-d} \cR$.
\end{theorem}

While this provides an algorithm to estimate the optimal {\em value} within a factor of $e^d$, via computing the fractional optimum $\cR$, it was not clear how to convert the fractional solution into a discrete solution of good value in polynomial time. Theorem~\ref{thm:NS-estimate} provides a guarantee in expectation, but this is not sufficient, since the expectation is exponentially small compared to $\cR$, and hence the probability of success of a direct randomized rounding method could be also exponentially small. This obstacle remained unresolved since 2016.

\paragraph{Conditional expectations.}
A natural idea is to apply the method of conditional expectations, and choose vectors $\bv_{i,j_i}$ one by one, in a way that avoids decreasing the quantity $D(\bx)$. Due to the product structure of $D(\bx)$, this method would work -- if we could estimate the conditional expectations given by $D(\bx)$ for partially rounded solutions. However, computing $D(\bx)$ (known as the {\em mixed-discriminant problem}) is \#P-hard and the known approximation algorithms are too weak for our purposes \cite{Gurvits09}.

Still, this idea has been useful in some special cases, notably in the case of Nash Social Welfare which can be formulated as a special case of the determinant maximization problem \cite{AGSS17}. The conditional expectations can be evaluated in this case, because they reduce to the problem of counting perfect matchings, which has an efficient approximation scheme. This gives an $e$-approximation for the Nash Social Welfare problem (after taking a $1/d$ power in the NSW objective).

\paragraph{Multilinear relaxation.}
Consider now a different, simpler relaxation:
\begin{equation}
\label{eq:vector-relax}
\det [\bw_1, \ldots, \bw_d] =  \det \left[ \sum_{j \in J_1} y_{1j} \bv_{1j}, \ldots, \sum_{j \in J_d} y_{dj} \bv_{dj} \right]
\end{equation}
where $\sum_{j \in J_i} |y_{ij}| = 1$. This relaxation is much easier to deal with in terms of rounding, because (1) this quantity can be easily computed, (2) by multilinearity of the determinant, we can pick the best vector $\bv_{ij}$ for each part, one by one, in a way that does not decrease the determinant. 

This formulation was used by \cite{ESV17}, in conjunction with other ideas, to derive the first approximation algorithm that finds a feasible solution. Their algorithm finds a random fractional solution of (\ref{eq:vector-relax}) and uses anticoncentration inequalities to prove a guarantee on its value. The approximation guarantee is roughly in the form $\prod_{i=1}^{m} |J_i|$; i.e. depending on the number of vectors in each part. In the regime where $|J_i| = \poly(m)$, this gives an $m^{O(m)}$-approximation. 

\paragraph{Sparsification approach.}
The first approximation algorithm based on the saddle-point relaxation (\ref{eq:saddle-point}) was developed by \cite{MNST20}. It relies on careful investigation of the structure of fractional solutions, and the technique of sparsification: reducing the support of the fractional solution to a size depending only on $d$. A randomized rounding method can be then applied to the sparsified solution to produce algorithmically a discrete solution of value at least $\OPT / d^{O(d^3)}$.  This method actually works for a more general problem than the one we are discussing here, under a general matroid constraint. For a partition matroid, the approximation factor improves to $e^{O(d^3)}$ (but is still far off from the desired $e^{O(d)}$).

\paragraph{Matroid intersection approach.}
Another approach, based on the combinatorial structure of matroid intersection was developed by \cite{BLPST22}. Matroid intersection is at the heart of determinant maximization subject to a matroid constraint, because even to find any nonzero solution, one has to solve the matroid intersection for the constraint matroid and the linear matroid given by the vectors. This approach, based on a further development of matroid intersection techniques, gave an algorithm to find a feasible solution within a factor of $m^{O(m)}$ for $m \leq d$ and $d^{O(d)}$ for $m > d$ \cite{BLPST22,BLPS22} (even for general matroids). This was until now the state-of-the-art algorithmic result for the problem we are discussing here. The approach is conceptually elegant; but it appears that a factor of $d!$ (or $m!$) is a fundamental barrier for this combinatorial method, (very roughly) due to the gap between the determinant and the maximum contribution of one of the $d!$ permutations.


\subsection{Intuition behind the new approach}


Our approach is based on the Nikolov-Singh relaxation (\ref{eq:saddle-point}) and a new randomized transformation technique which converts the fractional solution into the form $(\ref{eq:vector-relax})$. 
Since we would like to keep (\ref{eq:saddle-point}) as our starting point, and we already have a connection between (\ref{eq:saddle-point}) and (\ref{eq:multilinear}), we would like to establish a transformation that provides a link between (\ref{eq:multilinear}) and (\ref{eq:vector-relax}).

Given an optimal solution $(\bx, \bz)$ to (\ref{eq:saddle-point}), it doesn't work to just define $y_{ij} = x_{ij}$. For example, we could have $\bv_{ij} = \be_j$ for $j \in [d]$, where $\be_1,\ldots,\be_d$ is the standard basis in $\R^d$, and $x_{ij} = \frac{1}{d}, z_i = 0$ for all $i,j$.
Then we have $\sum_{i,j=1}^{d} x_{ij} \bv_{ij} \bv_{ij}^T = \sum_{j=1}^{d} \be_j \be_j^T = I$. However, replacing the vectors by $\bw_i = \sum_{j=1}^{d} y_{ij} \bv_{ij} = \frac{1}{d} \sum_{j=1}^{d} \be_j$ yields the same vector for each $i$ and the resulting determinant is $0$.

Note that in contrast, the quantity (\ref{eq:multilinear}) can be interpreted as selecting a vector for each part independently with probabilities $x_{ij}$. The problem with this is that this choice is too brittle; for example in the example above, it will pick a different vector from each part with probability only $\frac{d!}{d^d} \simeq e^{-d}$. Otherwise, the determinant is again $0$.

\paragraph{Randomized transformation from (\ref{eq:multilinear}) to (\ref{eq:vector-relax}).}
The key idea here is to design a {\em randomized linear combination} $\bw_i = \sum_{j \in J_i} Y_{ij} \bv_{ij}$, based on the fractional solution $x_{ij}$ in a smoother way, such that we can relate (\ref{eq:vector-relax}) to (\ref{eq:multilinear}) with good probability.
A natural property of this transformation is that $Y_{ij}$ should be distributed symmetrically around $0$: this is because $\bv_{ij}$ plays the same role in the solution as $-\bv_{ij}$, and in fact replacing $\bv_{ij}$ by $-\bv_{ij}$ does not affect the objectives (\ref{eq:saddle-point}), (\ref{eq:multilinear}). Hence intuitively we should not discriminate between any vector and its negative. For example, we could define $Y_{ij} = \pm x_{ij}$ where the signs are chosen independently at random. This would remove the deterministic coincidence in the basis example, but this alone
does not give a good value. For $Y_{ij}=\pm \frac{1}{d}$ and $\bw_i = \sum_{j=1}^{d} \pm \frac{1}{d} \be_j$, each column has norm $d^{-1/2}$, so even perfectly orthogonal columns would have absolute
determinant only $d^{-d/2}$.

Another useful intuition comes from splitting of fractional coefficients. Suppose we replace
a vector $\bv_{ij}$ with coefficient $x_{ij}$ by identical copies with coefficients
$\xi_1,\ldots,\xi_k$ that sum up to $x_{ij}$. Then the fractional contribution
$\sum_{\ell=1}^{k} \xi_\ell \bv_{ij} \bv_{ij}^T$ is equal to $x_{ij} \bv_{ij} \bv_{ij}^T$
and \eqref{eq:multilinear} is unchanged. Ideally, we would like the contribution to our transformed vector, $\sum_\ell Y_\ell \bv_{ij}$, to be also close to $x_{ij} \bv_{ij}$. But again, choosing $Y_\ell = \pm \xi_\ell$ does not work, since the summation $\sum Y_\ell$ concentrates around $0$ rather than the value $\sum_\ell \xi_\ell$. Other choices like Gaussian random variables $Y_\ell$ suffer from the same issue: for duplicated vectors with equally split coefficients, the norm deteriorates by a factor of $\sqrt{k}$. We need a symmetric distribution such that the signed sum remains comparable to the absolute sum, without significant cancellations. It might seem that there is no distribution with such properties, but there is a solution.

\paragraph{$1/2$-stable distributions.}
A symmetric, $\alpha$-stable distribution $\cD_\alpha$ is such that two independently random variables $Z_1, Z_2 \sim \cD_\alpha$ satisfy $c_1 Z_1 + c_2 Z_2 = c Z$ where $Z \sim \cD_\alpha$ and $c^\alpha = |c_1|^\alpha + |c_2|^\alpha$. In particular, the Gaussian distribution satisfies this with $\alpha=2$. 

It turns out that a better choice for our purpose is $\alpha=1/2$. (Other values of $\alpha \in (0,1)$ might also work, but the analysis is cleaner with $\alpha=1/2$.) A centrally symmetric $1/2$-stable distribution can be generated as follows:
$$ Z = \frac{1}{2X_1^2} - \frac{1}{2X_2^2} $$
where $X_1,X_2$ are independent standard Gaussian random variables. Let's call this distribution $\cD_{1/2}$. The distribution of $\frac{1}{2X_1^2}$ on its own is also $1/2$-stable (but not centrally symmetric); this is known as the L\'evy distribution \cite{Levy1925}.

$\cD_{1/2}$ has some interesting and counter-intuitive properties: For independent $Z_1, \ldots, Z_n \sim \cD_{1/2}$,
\begin{itemize}
\item $\frac{1}{n^2} Z_1 + \frac{1}{n^2} Z_2 + \ldots + \frac{1}{n^2} Z_n = Z$ has the same $\cD_{1/2}$ distribution (which cannot be true for a random variable with a finite $\E[|Z|]$; hence $\E[|Z|]=+\infty$);
\item the logarithmic expectation is finite: $\E[\log |Z|] = \kappa \simeq 1.27$;\footnote{All logarithms in this paper are natural.}
\item $\frac{1}{n^2} |Z_1| + \frac{1}{n^2} |Z_2| + \ldots + \frac{1}{n^2} |Z_n| = Z'$,
a different distribution supported on $\R_+$, but not a dramatically different logarithmic expectation:
\item $\E[\log Z'] \leq \kappa + \log 2 \simeq 1.96$.
\end{itemize}
Hence, performing a random walk according to these random variables, whether in absolute value or with random $\pm$ signs, results in a distribution which has roughly the same logarithmic expectation. This surprising property is at the core of why this method has a chance to succeed.

More generally, the linear combination $\sum_j c_j Z_j$ has the same distribution as $(\sum_j |c_j|^{1/2})^2 Z$. Therefore, starting with a fractional solution $x_{ij} \geq 0$ such that $\sum_{j \in J_i} x_{ij} = 1$, a natural way to define the transformation is $\bw_i = \frac{1}{\sum_{j \in J_i} |Y_{ij}|} \sum_{j \in J_i} Y_{ij} \bv_{ij}$ where $Y_{ij} = x_{ij}^2 Z_{ij}$. For example if $\bv_{ij} = \bv_i$ for all $j \in J_i$, we obtain $\bw_i = \frac{\sum_j x_{ij}^2 Z_{ij}}{\sum_j x_{ij}^2 |Z_{ij}|} \bv_i$, where $\sum_j x_{ij}^2 Z_{ij}$ is again distributed according to $\cD_{1/2}$. Also, the normalization coefficient $\sum_j x_{ij}^2 |Z_{ij}|$ has a comparable logarithmic expectation, and hence we don't lose a significant factor in this transformation, at least in terms of scaling. 

We record the crucial properties in the following lemma. These properties are well-known facts about stable distributions. For completeness, we justify them in Appendix~\ref{sec:stable-facts}.

\begin{lemma}
\label{lem:1/2-stable}
Let $\cD_{1/2}$ be the distribution of a random variable defined as $Z = \frac{1}{2X_1^2} - \frac{1}{2X_2^2}$, where $X_1,X_2$ are independent standard Gaussian variables. Let $Z_1,\ldots,Z_n$ be independent with distribution $\cD_{1/2}$ and coefficients $c_1,\ldots,c_n \in\R$ such that $\sum_{j=1}^{n} \sqrt{|c_j|} = \sqrt{c}$. Then,
\begin{itemize}
\item $\sum_{i=1}^{n} c_i Z_i \stackrel{\mathrm d}{=} c Z$ where $Z \sim \cD_{1/2}$. 
\item $\E[|\log |Z||]<\infty$ and $\E[\log |Z|] = \kappa \simeq 1.27$.
\item Let $Z'=\frac{1}{c}\sum_{i=1}^{n} |c_i|\cdot |Z_i|$. Then, $\E[|\log Z'|]<\infty$ and $\E[\log Z'] \leq \kappa + \log 2$.
\end{itemize}
\end{lemma}

We summarize our high-level plan: Solve the Nikolov-Singh relaxation (\ref{eq:saddle-point}), obtain a fractional solution $\bx$, then convert it into a solution of the multilinear relaxation (\ref{eq:vector-relax}) via $1/2$-stable random variables, and finally round the relaxation (\ref{eq:vector-relax}) via conditional expectations. This still does not explain why the analysis works out for general instances. We turn to the analysis in the next section.

\section{Algorithm and analysis for $m=d$}

Let us now describe and analyze our algorithm in the $m=d$ case. The algorithm starts by solving the Nikolov-Singh relaxation, which yields a fractional solution $(\bx,\bz)$. We denote $$ R(\bx, \bz) = \det \left( \sum_{i=1}^{d} \sum_{j \in J_i} e^{z_i} x_{ij} \bv_{ij} \bv_{ij}^T \right).$$ 
We remark that in general, we cannot find an exact optimal solution of (\ref{eq:saddle-point}), and in fact it might not even exist, due to an infimum taken over an unbounded set. But we can find a near-optimal solution, by which we mean a solution achieving $R(\bx,\bz) \leq (1+\epsilon) R(\bx) = (1+\epsilon) \inf_\bz R(\bx,\bz)$ and $R(\bx) \geq (1-\epsilon) \cR$. Furthermore, the related quantity $D(\bx)$ satisfies $D(\bx) \geq \frac{d!}{d^d} R(\bx) \geq e^{-d} \cR$ \cite{NS16}.
Given such a fractional solution $(\bx,\bz)$, we execute the randomized rounding algorithm given in Algorithm 1. The rest of the analysis is concerned with Algorithm 1.

\begin{algorithm}[H]
\caption{Randomized rounding for a simple partition matroid, $m=d$}
\label{alg:detmax-rounding}
\begin{algorithmic}[1]
\Require Vectors $\bv_{ij} \in \R^d$, a near-optimal fractional solution $(\bx,\bz)$ and a bound $\cR$.
\Ensure One index $j(i) \in J_i$ chosen for each $1 \leq i \leq d$.
\Repeat
\State Draw independent random variables $Z_{ij} \sim \cD_{1/2}$ for $1 \leq i \leq d$, $j \in J_i$
\State $y_{ij} \gets x_{ij}^2 Z_{ij} / \sum_{k \in J_i} x_{ik}^2 |Z_{ik}|$ 
\State $\bw_i \gets \sum_{j \in J_i} y_{ij} \bv_{ij}$
\Until $\det^2 [\bw_1,\ldots,\bw_d] \geq e^{-6d} \cR$.
\For{$i=1,\ldots,d$}
    \State Choose $j(i) \in J_i$ to maximize 
    $|\det [\bv_{1,j(1)},\ldots,\bv_{i,j(i)},\bw_{i+1},\ldots,\bw_d]|$.
\EndFor
\State \Return $(j(1),\ldots,j(d))$.
\end{algorithmic}
\end{algorithm}


It's very easy to prove that if this algorithm terminates, it provides an $e^{-6d}$-approximation.

\begin{lemma}
If Algorithm 1 terminates, it returns a solution such that 
$$ (\det [\bv_{1,j(1)},\ldots,\bv_{d,j(d)}])^2 \geq e^{-6d} \cR.$$
\end{lemma}

\begin{proof}
If the algorithm gets to Step 6, we have $|\det [\bw_1,\ldots,\bw_d]| \geq e^{-3d} \sqrt{\cR}$. Now observe that for each $i$, $\bw_i = \sum_{j \in J_i} y_{ij} \bv_{ij}$, and $\sum_{j \in J_i} |y_{ij}| = 1$.
By the multilinearity of the determinant, we can choose $j(i)$ such that replacing $\bw_i$ by $\bv_{i,j(i)}$ does not decrease the determinant in absolute value. Hence, one by one the algorithm replaces the vectors $\bw_1,\ldots,\bw_d$ by $\bv_{1,j(1)}, \ldots, \bv_{d,j(d)}$ so that
$$ |\det [\bv_{1,j(1)}, \ldots, \bv_{d,j(d)}]| \geq |\det [\bw_1,\ldots,\bw_d]| \geq e^{-3d} \sqrt{\cR}.$$
\end{proof}

It remains to prove that the sampling procedure in Steps 2-5 succeeds with finding a good choice of vectors $\bw_1,\ldots,\bw_d$ in an expected polynomial number of steps.

\subsection{Analysis of the sampling procedure}

Our remaining goal is to prove that the algorithm will succeed in Steps 2-5 with constant probability, and hence it will terminate in a constant expected number of iterations. This is captured by the following.

\begin{theorem}
\label{thm:stable-chaos}
Consider a fractional solution $(\bx,\bz)$ of the relaxation (\ref{eq:saddle-point}), satisfying $D(\bx) \geq e^{-d} \cR$.
Let $\bw_1,\ldots,\bw_d$ be the random vectors constructed by the algorithm. Then
with constant probability, $$ (\det [\bw_1,\ldots,\bw_d])^2 \geq e^{-6d} \cR.$$
\end{theorem}

The proof of Theorem~\ref{thm:stable-chaos} boils down to a comparison between two quantities:
\begin{equation}
S(\bx) := \sum_{j_1 \in J_1,\ldots,j_d \in J_d} \left( \prod_{i=1}^{d} x_{i,j_i} \right) |\det [\bv_{1,j_1}, \ldots, \bv_{d,j_d}]|^{1/2}
\end{equation}
and 
\begin{equation}
Q(\bx) := \sum_{j_1 \in J_1,\ldots,j_d \in J_d} \left( \prod_{i=1}^{d} x^2_{i,j_i} Z_{i,j_i} \right) \det [\bv_{1,j_1}, \ldots, \bv_{d,j_d}]
\end{equation}
The (deterministic) quantity $S(\bx)$ is related to $\cR$ as follows:
Recall that $$D(\bx) = \sum_{j_1 \in J_1,\ldots,j_d \in J_d} \left( \prod_{i=1}^{d} x_{i,j_i} \right) (\det [\bv_{1,j_1}, \ldots, \bv_{d,j_d}])^2$$ satisfies $D(\bx) \geq e^{-d} \cR$. Also, $(\det [\bv_{1,j_1}, \ldots, \bv_{d,j_d}])^2 \leq \cR$ for any choice of $j_1,\ldots,j_d$. Hence we have $\cR^{1/4} \geq S(\bx) \geq D(\bx) / \cR^{3/4} \geq e^{-d} \cR^{1/4}$. 

The (random) quantity $Q(\bx)$ describes the performance of a particular sample $\bw_1,\ldots,\bw_d$. 
Recall that $$\bw_i = \sum_{j \in J_i} y_{ij} \bv_{ij} = \frac{\sum_{j \in J_i} x_{ij}^2 Z_{ij} \bv_{ij}}{\sum_{j \in J_i} x_{ij}^2 |Z_{ij}|}.$$
Hence,
$$ \det [\bw_1,\ldots,\bw_d] = \frac{\det \left[ \sum_{j \in J_1} x_{1j}^2 Z_{1j} \bv_{1j}, \ldots, \sum_{j \in J_d} x_{dj}^2 Z_{dj} \bv_{d j} \right]}{ \prod_{i=1}^{d} \sum_{j \in J_i} x_{ij}^2 |Z_{ij}|}
= \frac{Q(\bx)}{ \prod_{i=1}^{d} \sum_{j \in J_i} x_{ij}^2 |Z_{ij}|}.$$
Our goal in the following will be to estimate
$$ \E[\log |\det [\bw_1,\ldots,\bw_d]|] = \E[\log |Q(\bx)|] - \sum_{i=1}^{d} \E[\log \sum_{j \in J_i} x_{ij}^2 |Z_{ij}|]. $$
The last summation is easy to estimate by Lemma~\ref{lem:1/2-stable}: Since $\sum_{j \in J_i} x_{ij} = 1$, we have $\sum_{i=1}^{d} \E[\log $ $\sum_{j \in J_i} x_{ij}^2 |Z_{ij}|]$ $\leq d(\kappa + \log 2).$
Hence,
$$ \E[\log |\det [\bw_1,\ldots,\bw_d]|] \geq \E[\log |Q(\bx)|] - d (\kappa + \log 2). $$
We claim that 
\begin{equation}
\label{eq:stable-chaos}
\E[\log |Q(\bx)|] \geq 2 \log S(\bx) + d \kappa.
\end{equation}
If this is the case, then  
$$ \E[\log |\det [\bw_1,\ldots,\bw_d]|] \geq 2 \log S(\bx) - d \log 2
 \geq \frac12 \log \cR - 2d - d \log 2 $$
and since $\log |\det [\bw_1,\ldots,\bw_d]| \leq \frac12 \log \cR$ always,
we obtain $\log |\det [\bw_1,\ldots,\bw_d]| \geq \frac12 \log \cR - 3d$ with probability at least $1 - \frac{2+\log 2}{3} = \frac{1-\log 2}{3}$ by Markov's inequality.

In fact, to prove Theorem~\ref{thm:stable-chaos} with a weaker bound, it is sufficient to show that $\E[\log |Q(\bx)|] \geq 2 \log S(\bx) - O(d)$. This would only result in a worse constant in the exponent. We prove the inequality (\ref{eq:stable-chaos}) above.

\subsection{The ``logarithmic stable chaos'' inequality}

The comparison between $S(\bx)$ and $Q(\bx)$ is accomplished by an inductive argument, unrolling the expectation over the product distribution defined by $\bx$, one block of the partition at a time. We first set up some notation and explain the high-level intuition. 

Consider a tensor $a_\bj$ indexed by $\bj = (j_1,\ldots,j_m)$ where $j_i \in J_i$.
Our choice for $m=d$ will be $a_\bj = \det [\bv_{1 j_1}, \bv_{2 j_2}, \ldots \bv_{d j_d}]$,
but the inequality in this section holds for {\em any} $a_\bj$. Note also that in this section we do not require $m=d$; the statement actually does not involve the parameter $d$.

Consistently with our notation above, we set 
$$ S(\bx) := \sum_{\bj \in J_1 \times \ldots \times J_m} \left( \prod_{i=1}^{m} x_{i,j_i} \right) |a_\bj|^{1/2} $$
and
$$ Q(\bx) := \sum_{\bj \in J_1 \times \ldots \times J_m} \left( \prod_{i=1}^{m} x^2_{i,j_i} Z_{i,j_i} \right) a_\bj.$$
Recall that $Z_{i,j}$ are independent random variables with distribution $\cD_{1/2}$ and $\E[\log |Z_{i,j}|] = \kappa$.
We claim the following inequality\footnote{The term ``logarithmic stable chaos'' was suggested by ChatGPT, by analogy with the study of random systems in the form of a multilinear polynomial in independent random variables, which is often called a homogeneous or polynomial chaos.
There is no connection with chaos theory.} which we earlier called (\ref{eq:stable-chaos}). 

\begin{lemma}[Logarithmic stable chaos]
\label{lem:stable-chaos}
Assuming that $S(\bx)>0$,
$$ \E[\log |Q(\bx)|] \geq 2 \log S(\bx) + m \kappa.$$
\end{lemma}

For $m=1$, this holds due to the basic property of $1/2$-stable random variables ($\sum_j x_j^2 Z_j$ being distributed as $(\sum |x_j|) Z$). For $m>1$, however, this is a much more intricate statement. 
An interpretation of this inequality is that $Q(\bx)$, which is a random linear combination of the {\em signed} contributions $a_\bj$, is a good estimate of $S(\bx)$ which is a related linear combination of the {\em absolute values} $|a_\bj|^{1/2}$. The presence of a square root in $S(\bx)$ is an artifact of the use of $1/2$-stable distributions, and as we already argued, it is sufficient for our purposes.

\paragraph{The Gibbs inequality.}
A standard inequality we employ here is the following.
$D_{KL}$ denotes the KL-divergence, $D_{KL}(q \| p) = \sum_j q_j \log \frac{q_j}{p_j}$.
The inequality is essentially just the concavity of the logarithm.

\begin{proposition}[Gibbs inequality]
\label{prop:Gibbs}
For probability distributions $p$ and $q$ where $\mbox{supp}(q) \subseteq \mbox{supp}(p)$, and values $\gamma_j>0$ on $\mbox{supp}(p)$,
$$ \log \sum_{j \in supp(p)} p_j \gamma_j \geq \sum_{j \in supp(q)} q_j \log \gamma_j - D_{KL}(q \| p).$$
\end{proposition}

\begin{proof}
$\log \sum_{j \in supp(p)} p_j \gamma_j \geq \log \sum_{j \in supp(q)} q_j \frac{p_j \gamma_j}{q_j} 
\geq \sum_{j \in supp(q)} q_j \log \frac{p_j \gamma_j}{q_j} = \sum_j q_j \log \gamma_j - D_{KL}(q \| p)$.
\end{proof}

Now we can proceed to the proof of Lemma~\ref{lem:stable-chaos}.

\begin{proof}[Proof of Lemma~\ref{lem:stable-chaos}]
We proceed by induction on $m$.
Let us first consider the case of $m=1$. Here we have $S(\bx) = \sum_{j \in J_1} x_{1j} |a_j|^{1/2}$ and $Q(\bx) = \sum_{j \in J_1} x_{1j}^2 Z_{1j} a_j$ (a linear combination of independent $1/2$-stable random variables).  By Lemma~\ref{lem:1/2-stable}, $Q(\bx) \stackrel{\mathrm d}{=} (\sum_{j \in J_1} x_{1j} |a_j|^{1/2})^2 Z$ where $Z$ is another $1/2$-stable random variable. Hence,
$$ \E[\log |Q(\bx)|] = 2 \log \left(\sum_{j \in J_1} x_{1j} |a_j|^{1/2} \right) +  \E[\log |Z|]
= 2 \log S(\bx) + \kappa.$$

Next, we consider $m>1$. We unroll the definition of $Q(\bx)$ as follows:
$$ Q(\bx) = \sum_{j_1 \in J_1,\ldots,j_m \in J_m} \left( \prod_{i=1}^{m} x^2_{i,j_i} Z_{i,j_i} \right) a_\bj
= \sum_{j_1 \in J_1} x_{1,j_1}^2 Z_{1,j_1} Q_{j_1}(\bx') $$
where $\bx' = (x_{ij})_{i \geq 2, j \in J_i}$ and $Q_{j_1}(\bx') = \sum_{j_2,\ldots,j_m} \left( \prod_{i=2}^{m} x^2_{i,j_i} Z_{i,j_i} \right) a_{(j_1,j_2,\ldots,j_m)}$. $Q_{j_1}(\bx')$ is a quantity of the same type as $Q(\bx)$, on $m-1$ parts (and hence our goal is to invoke induction on $Q_{j_1}(\bx'))$. We also define analogously $S_{j_1}(\bx') = \sum_{j_2,\ldots,j_m} \prod_{i=2}^{m} x_{i,j_i} |a_{(j_1,j_2,\ldots,j_m)}|^{1/2}$.

\paragraph{Logarithmic integrability.}
Our proofs refer to a number of logarithmic expectations. Let us argue first that these are well-defined and finite. More precisely, we claim inductively that unless the process is identically $0$, $Q(\bx) \neq 0$ almost surely, 
and $\E[|\log |Q(\bx)||] < \infty$ for every integer $m$,
fractional solution $\bx$ and real tensor $a_\bj$.

This is true for $m=1$, by the basic properties of $\cD_{1/2}$ (Lemma~\ref{lem:1/2-stable}). For $m>1$, let us assume inductively that $Q_{j_1}(\bx) \neq 0$ almost surely and $\E[|\log |Q_{j_1}(\bx')||] < \infty$ for each $j_1 \in J_1$. (We can ignore the contributions which are identically $0$.) We have $Q(\bx) = \sum_{j_1 \in J_1} x_{1,j_1}^2 Z_{1,j_1} Q_{j_1}(\bx')$. For fixed values of $Q_{j_1}(\bx')$, $Q(\bx)$ is distributed as $(\sum_{j_1 \in J_1} x_{1,j_1} |Q_{j_1}(\bx')|^{1/2})^2 Z$ where $Z \sim \cD_{1/2}$. Since $x_{1,j_1} \neq 0$ for at least one $j_1 \in J_1$, and $Q_{j_1}(\bx') \neq 0$ as well as $Z \neq 0$ almost surely, we have $Q(\bx) \neq 0$ almost surely. We estimate $|\log |Q(\bx)||$ using the elementary inequality
$|\log \sum_{j \in J} c_j| \leq \log |J| + \max_{j \in J} |\log c_j|$ for $c_j>0$, to obtain
\begin{align*}
\Big|\log |Q(\bx)|\Big| & =  \Big| 2 \log \sum_{j_1 \in J_1: x_{1,j_1} \neq 0} x_{1,j_1} |Q_{j_1}(\bx')|^{1/2} + \log |Z| \Big| \\
& \leq 2 \log |J_1| + 2\max_{j_1; x_{1,j_1} \neq 0} \left( |\log x_{1,j_1}| + \frac12 \Big|\log |Q_{j_1}(\bx')|\Big| \right) + \Big|\log |Z|\Big|.
\end{align*}
As $\E[|\log |Q_{j_1}(\bx')||] < \infty$ by assumption, we also have $\E[|\log |Q(\bx)||] < \infty$.

\paragraph{Induction.}
Recall the quantities $Q_{j_1}(\bx')$ and $S_{j_1}(\bx')$ defined above.
First, let us condition on all the variables $Z_{ij}, i \geq 2$, which makes $Q_{j_1}(\bx')$ deterministic and $Q(\bx)$ is now a linear combination of the random variables $Z_{1,j_1}$.
We use Lemma~\ref{lem:1/2-stable} to compute
$$ \E[\log |Q(\bx)| \mid Z_{i,j}: i \geq 2] 
= \E[\log \sum_{j_1 \in J_1} x_{1,j_1}^2 Q_{j_1}(\bx') \,  Z_{1,j_1}]
= 2 \log \sum_{j_1 \in J_1} x_{1,j_1} |Q_{j_1}(\bx')|^{1/2} + \kappa.$$
The goal in the following is to estimate this quantity. Here is where the Gibbs inequality comes in: We define $\gamma_j = |Q_j(\bx')|^{1/2}$, $p_j = x_{1,j}$ and 
$ q_j = x_{1j} \frac{S_j(\bx')}{S(\bx)} $.
The Gibbs inequality (interpreting terms where $S_{j_1}(\bx')=0$ as $0$) gives
$$ \log \sum_{j_1 \in J_1} x_{1,j_1} |Q_{j_1}(\bx')|^{1/2}
= \log \sum_{j_1 \in J_1} p_{j_1} \gamma_{j_1} \geq \sum_{j_1 \in J_1} q_{j_1} \log \gamma_{j_1} - D_{KL}(q \| p) $$
$$ = \frac{1}{2 S(\bx)} \sum_{j_1 \in J_1} x_{1 j_1} S_{j_1}(\bx') \log |Q_{j_1}(\bx')| - D_{KL}(q \| p).$$
Taking an expectation over $Z_{ij}, i \geq 2$, 
$$ \E\left[ \log \sum_{j_1 \in J_1} x_{1,j_1} |Q_{j_1}(\bx')|^{1/2} \right] 
 \geq \frac{1}{2 S(\bx)} \sum_{j_1 \in J_1} x_{1 j_1} S_{j_1}(\bx') \, \E[\log |Q_{j_1}(\bx')|] -  D_{KL}(q \| p).$$
Here we apply the inductive hypothesis: for each $j_1 \in J_1$ where $S_{j_1}(\bx') > 0$,
$ \E[\log |Q_{j_1}(\bx')|] \geq 2 \log S_{j_1}(\bx') + (m-1) \kappa$
and hence 
$$ \E\left[ \log \sum_{j_1 \in J_1} x_{1,j_1} |Q_{j_1}(\bx')|^{1/2} \right] 
 \geq \frac{1}{S(\bx)} \sum_{j_1 \in J_1} x_{1 j_1} S_{j_1}(\bx') \, \log S_{j_1}(\bx')
 + (m-1) \kappa -  D_{KL}(q \| p).$$
We used the fact that by definition, $S(\bx) = \sum_{j_1 \in J_1} x_{1 j_1} S_{j_1}(\bx')$.
Also, the KL-divergence can be computed explicitly as follows:
$$ D_{KL}(q \| p) =  \sum_{j \in J_1} q_j \log \frac{q_j}{p_j}
= \frac{1}{S(\bx)}  \sum_{j_1 \in J_1} x_{1 j_1} S_{j_1}(\bx') \log \frac{S_{j_1}(\bx')}{S(\bx)} $$
$$ = \frac{1}{S(\bx)} \sum_{j_1 \in J_1} x_{1 j_1} S_{j_1}(\bx') \log S_{j_1}(\bx')
- \log S(\bx).$$
Putting all this together, we obtain
$$ \E[\log |Q(\bx)|] 
= 2 \E\left[ \log \sum_{j_1 \in J_1} x_{1,j_1} |Q_{j_1}(\bx')|^{1/2} \right] + \kappa
\geq 2 \log S(\bx) + m \kappa.$$
\end{proof}
This completes the proof of Lemma~\ref{lem:stable-chaos}, and hence Theorem~\ref{thm:stable-chaos}, which implies our main result for $m=d$.

\section{Algorithm and analysis for $m \geq d$}

Now let us consider a more general setting, where $\bv_{ij} \in \R^d$ we are choosing one index from each of $m$ sets $J_1, J_2,\ldots, J_m$, $m \geq d$, in order to maximize
$$ \det \left( \sum_{i=1}^{m} \bv_{i,j_i} \bv_{i,j_i}^T \right).$$
In the case of $m=d$, this coincides with the objective $\det^2 [\bv_{1,j_1}, \ldots, \bv_{d,j_d}]$ discussed above. A natural extension of the approach presented so far gives a factor of $e^{O(m)}$. 
We aim to prove Theorem~\ref{thm:main} with a factor of $e^{O(d)}$, which presents additional difficulties.

An appropriate relaxation for the regime of $m \geq d$ was proposed in \cite{MNST20}. We present a simplified formulation here for the case of partition matroids:
\begin{eqnarray*}
\cR =  \max_\bx \inf_\bz & \det \left (\sum_{i=1}^{m} \sum_{j \in J_i} e^{z_i} x_{ij} \bv_{ij} \bv_{ij}^T \right) \\
\forall i \in [m]; & \sum_{j \in J_i} x_{ij} = 1, \\
\forall i \in [m], j \in J_i; & x_{ij} \geq 0, \\
\forall I \subseteq [m], |I|=d; & \sum_{i \in I} z_i \geq 0.
\end{eqnarray*}

Note the main change compared to Section~\ref{sec:prior}, which is that the constraint on $z_i$ is applied to all subsets of $d$ variables.
Again, it is relatively straightforward to check that $\cR \geq \OPT$. 

\begin{algorithm}[H]
\caption{Randomized rounding for a simple partition matroid, $m \geq d$}
\label{alg:detmax-rounding}
\begin{algorithmic}[1]
\Require Vectors $\bv_{ij} \in \R^d$, a near-optimal fractional solution $(\bx,\bz)$ and a bound $\cR$.
\Ensure One index $j(i) \in J_i$ chosen for each $1 \leq i \leq m$.
\Repeat
\State Draw independent random variables $Z_{ij} \sim \cD_{1/2}$ for $1 \leq i \leq m$, $j \in J_i$
\State $y_{ij} \gets x_{ij}^2 Z_{ij} / \sum_{j \in J_i} x_{ij}^2 |Z_{ij}|$ 
\State $\bw_i \gets \sum_{j \in J_i} y_{ij} \bv_{ij}$
\Until $\det (\sum_{i=1}^{m} \bw_i \bw_i^T) \geq e^{-9d} \cR$.
\For{$\ell=1,\ldots,m$}
    \State Choose $j(\ell) \in J_\ell$ to maximize 
    $\det \left(\sum_{i=1}^{\ell} \bv_{i,j(i)} \bv_{i,j(i)}^T + \sum_{i=\ell+1}^{m} \bw_i \bw_i^T \right)$.
\EndFor
\State \Return $(j(1),\ldots,j(m))$.
\end{algorithmic}
\end{algorithm}

We justify first that the algorithm succeeds if it gets to Step 6. This is slightly less obvious here than it was in the case of $m=d$.

\begin{lemma}
\label{lem:rounding-m}
Given vectors $\bw_1,\ldots,\bw_m$, the rounding procedure in Steps 6-7 finds vectors $\bv_{1,j(1)}, \ldots,$ $\bv_{m,j(m)}$ such that
$$ \det \left(\sum_{i=1}^{m} \bv_{i,j(i)} \bv_{i,j(i)}^T \right) \geq \det \left(\sum_{i=1}^{m} \bw_i \bw_i^T \right).$$
\end{lemma}

\begin{proof}
Denoting the vectors after $\ell-1$ steps of the rounding procedure by $\bw'_1 = \bv_{1,j(1)}, \ldots, \bw'_{\ell-1} = \bv_{\ell-1,j(\ell-1)}, \bw'_{\ell} = \bw_{\ell}, \ldots, \bw'_m = \bw_m$.
We can use again the Cauchy-Binet expansion: 
$$ \det \left( \sum_{i=1}^{m} \bw'_i (\bw'_i)^T \right)
= \sum_{\substack{T \subseteq [m] \\ |T|=d}} (\det [\bw'_i: i \in T])^2. $$
Fixing $\ell$ and expanding $\bw'_{\ell} = \bw_{\ell} = \sum_{j \in J_{\ell}} y_{\ell j} \bv_{\ell j}$, with $\sum_{j \in J_\ell} |y_{\ell j}| = 1$, we see that 
$\det [\bw'_i: i \in T]$ is a linear function in $y_{\ell j}, j \in J_\ell$ for the sets $T$ containing $\ell$. Consequently, $(\det [ \bw'_i: i \in T])^2$ is a convex function in $y_{\ell j}$, and so is also the summation over $T$. Therefore, the maximum over $\sum_{j \in J_\ell} |y_{\ell j}| = 1$ is equal to the maximum over $\sum_{j \in J_\ell} |y_{\ell j}| \leq 1$, and attained at one of the vertices, $y_{\ell j(\ell)} = \pm 1$ and $y_{\ell j} = 0$ for $j \neq j(\ell)$. This means we can replace $\bw_\ell$ by $\bv_{\ell j(\ell)}$ so that the value of $\det \left( \sum_{i=1}^{m} \bw'_i (\bw'_i)^T \right)$ does not decrease.
\end{proof}

Hence it remains to analyze the value provided by the vectors $\bw_1,\ldots,\bw_m$.

\subsection{Analysis of the sampling procedure}

We denote again $R(\bx,\bz) = \det \left (\sum_{i=1}^{m} \sum_{j \in J_i} e^{z_i} x_{ij} \bv_{ij} \bv_{ij}^T \right)$, $R(\bx) = \inf_\bz R(\bx,\bz)$, $\cR = \max_\bx R(\bx)$, and 
$$ D(\bx) = \sum_{j_1 \in J_1,\ldots, j_m \in J_m} \left( \prod_{i=1}^{m} x_{i,j_i} \right) \det \left( \sum_{i=1}^{m} \bv_{i,j_i} \bv_{i,j_i}^T \right).$$
It follows from the techniques of \cite{NS16,AGSS17,MNST20} that $D(\bx^*) \geq \frac{m! (m-d)^{m-d}}{m^m (m-d)!} \cR > e^{-d} \cR$ for an optimal fractional solution $(\bx^*,\bz^*)$. Our algorithm finds a near-optimal solution $(\bx,\bz)$ of value $R(\bx) \geq (1-\epsilon) \cR$, and $D(\bx) \geq e^{-d} \cR$. 
Since an integrality gap bound of $e^d$ is not explicitly stated in these papers (for determinant maximization and $m>d$), we provide a proof in Appendix~\ref{sec:integrality-gap}.

Using the Cauchy-Binet formula, we can rewrite $D(\bx)$ as a summation over $d \times d$ determinants:
$$ D(\bx) = \sum_{j_1 \in J_1,\ldots, j_m \in J_m} \left( \prod_{i=1}^{m} x_{i,j_i} \right)
\sum_{\substack{T \subseteq [m] \\ |T|=d}} \det \left( \sum_{i \in T} \bv_{i,j_i} \bv_{i,j_i}^T \right) $$
$$ = \sum_{\substack{T \subseteq [m] \\ |T|=d}} \sum_{\substack{j_i \in J_i \\ \forall i \in T}} \left( \prod_{i \in T} x_{i,j_i} \right) \det \left( \sum_{i \in T} \bv_{i,j_i} \bv_{i,j_i}^T \right)
= \sum_{\substack{T \subseteq [m] \\ |T|=d}} D_T(\bx) $$
where $$D_T(\bx) = \sum_{\substack{j_i \in J_i \\ \forall i \in T}} \left( \prod_{i \in T} x_{i,j_i} \right) \det \left( \sum_{i \in T} \bv_{i,j_i} \bv_{i,j_i}^T \right).$$
Analogously, let us also define
$$ S_T(\bx) := \sum_{\substack{j_i \in J_i \\ \forall i \in T}} \left( \prod_{i \in T} x_{i,j_i} \right) \left(\det \left( \sum_{i \in T} \bv_{i,j_i} \bv_{i,j_i}^T \right)\right)^{1/4} $$
and
$$ M_T = \max_{\substack{j_i \in J_i \\ \forall i \in T}} \det \left( \sum_{i \in T} \bv_{i,j_i} \bv_{i,j_i}^T \right).$$
The analysis for a given subset $T$ of $d$ indices is similar to the earlier analysis for $m=d$: We get $S_T(\bx) \geq M_T^{-3/4} D_T(\bx)$, and the performance of our algorithm on the subset $T$ can be related to $S_T(\bx)$. More precisely, $\det (\sum_{i \in T} \bw_i \bw_i^T)$ for the vectors $\bw_i$ found by our algorithm can be related to $S_T^4(\bx)$, and hence the global determinant, $\det (\sum_{i=1}^{m} \bw_i \bw_i^T)$, can be related to $\sum_{T \subseteq [m]; |T|=d} S_T^4(\bx)$.

Since $S_T(\bx) \geq M_T^{-3/4} D_T(\bx)$, we obtain $\sum_{T \subseteq [m]; |T|=d} S_T^4(\bx) \geq \sum_{T \subseteq [m]; |T|=d} M_T^{-3} D_T^4(\bx)$. However, this bound could be potentially weak, if the main contributions $D_T(\bx)$ come primarily from sets $T$ where $M_T$ is very large. We need a bound showing that $\cR$ not only upper-bounds the global optimum, but it upper-bounds the $T$-restricted optima $M_T$ in a certain structured way. This is accomplished by the following.

\subsection{The leverage score bound} 

Given an optimal solution $(\bx,\bz)$ with $A(\bx,\bz) =  \sum_{i=1}^{m} \sum_{j \in J_i} e^{z_i} x_{ij} \bv_{ij} \bv_{ij}^T$, we can associate a ``leverage score'' $\bv_{ij}^T A^{-1}(\bx,\bz) \bv_{ij}$ with each vector $\bv_{ij}$. The leverage score measures how much the vector $\bv_{ij}$ affects the fractional solution at the saddle point. By KKT conditions, the contribution of each participating vector in a given part $i$ should be the same, $\mu_i$. It turns out that these parameters provide a bound on every $T$-restricted optimum $M_T$. Some technicalities arise since a perfect saddle point solution might not exist. It is sufficient to use a certain near-optimal solution; we remark that the technical conditions here are slightly different from what we need in the algorithm, and this solution is used only for the purpose of proving Lemma~\ref{lem:KKT}.


\begin{lemma}
\label{lem:KKT}
There are parameters $\mu_1,\ldots,\mu_m \geq 0$ such that $\sum_{i=1}^{m} \mu_i = d$, and for all $T \subseteq [m], |T|=d$:
$$M_T = \max_{\substack{j_i \in J_i \\ \forall i \in T}} \det \left(\sum_{i \in T} \bv_{i,j_i} \bv_{i,j_i}^T \right) \leq \cR \prod_{i \in T} \mu_i.$$
\end{lemma}

\begin{proof}
Let us define 
$$A(\bx,\bz) = \sum_{i=1}^{m} \sum_{j \in J_i} e^{z_i} x_{ij} \bv_{ij} \bv_{ij}^T.$$
I.e., we have $\cR = \max_\bx \inf_\bz \det A(\bx,\bz)$ subject to the constraints $\bx \geq 0, \forall i; \sum_{j \in J_i} x_{ij} = 1$ and $\forall T \subseteq [m], |T|=d; \sum_{j \in T} z_j \geq 0$.
$A(\bx,\bz)$ is a positive semidefinite $d \times d$ matrix. 
We can apply Sion's minimax theorem to the function $(\det A(\bx,\bz))^{1/d}$ which is concave in $\bx$ (by concavity of $(\det A)^{1/d}$ on the positive-semidefinite cone) and convex in $\bz$ (by convexity of $\log \det A(\bx,\bz)$ in $\bz$). Hence, we have $\max_\bx \inf_\bz (\det A(\bx,\bz))^{1/d} = \inf_\bz \max_\bx (\det A(\bx,\bz))^{1/d}$. In fact, the $1/d$ power can be removed in this statement due to monotonicity. We can choose a near-optimal $\bz^*$ and then a maximizer $\bx^*$ for this given $\bz^*$. Note that the maximization over $\bx$ is over a compact set, and hence an exact optimum exists. We obtain a solution $(\bx^*,\bz^*)$ such that
\begin{itemize}
\item $\cR \leq \det A(\bx^*,\bz^*) \leq (1+\epsilon) \cR$, i.e.~$(\bx^*,\bz^*)$ is near-optimal;
\item $\det A(\bx^*, \bz^*) \geq \det A(\bx, \bz^*)$ for every feasible $\bx$; i.e, $\bx^*$ is optimal for the given $\bz^*$.
\end{itemize}

Assuming that the optimization problem is non-trivial, $\log \det A(\bx,\bz)$ is well defined in the region where $\det A(\bx,\bz) > 0$ and $A(\bx,\bz)$ is invertible.
We use the KKT conditions for $\log \det A(\bx,\bz)$ with respect to the $x_{ij}$ variables: 
There are Lagrange multipliers $\mu_i \geq 0$ such that 
$$ \frac{\partial (\log \det A)}{\partial x_{ij}}(\bx^*,\bz^*) = e^{z^*_i} \bv_{ij}^T A^{-1}(\bx^*,\bz^*) \bv_{ij} \leq \mu_i $$
with equality whenever $x_{ij} > 0$. (The partial derivative is computed using the rank-one update formula.) In the following, we drop the arguments $\bx^*,\bz^*$. Using the trace rotation formula, we derive
$$ \sum_{i=1}^{m} \mu_i = \sum_{i=1}^{m} \sum_{j \in J_i} x^*_{ij} \mu_i
= \sum_{i=1}^{m} \sum_{j \in J_i} e^{z^*_i} x^*_{ij} \bv_{ij}^T A^{-1} \bv_{ij}
= \sum_{i=1}^{m} \sum_{j \in J_i} e^{z^*_i} x^*_{ij} Tr(\bv_{ij} \bv_{ij}^T A^{-1})
= Tr(A A^{-1}) = d.$$
Now, for any $T \subseteq [m], |T|=d$ and a choice of indices $(j_i: i \in T)$,
let $V$ be a matrix with columns $(\bv_{ij_i}: i \in T)$.
We have
$$ \det \left( \sum_{i \in T} \bv_{ij_i} \bv_{ij_i}^T \right) = \det (V V^T)
 = \det ( A V^T A^{-1} V) = \det A \cdot \det (V^T A^{-1/2}) \cdot \det (A^{-1/2} V) $$
 $$ \leq \det A \cdot \prod_{i \in T} \| A^{-1/2} \bv_{ij} \|^2 = \det A \cdot \prod_{i \in T} (\bv_{ij_i}^T A^{-1} \bv_{ij_i}) \leq \det A \cdot \prod_{i \in T} \left(e^{-z^*_i} \mu_i \right) \leq (1+\epsilon) \cR \prod_{i \in T} \mu_i.$$
We used the fact that since $A$ is positive definite, we can decompose $A^{-1} = A^{-1/2} A^{-1/2}$ and $\det (A^{-1/2} V) = \det (V^T A^{-1/2}) \leq \prod_{i \in T} \| A^{-1/2} \bv_{ij}\|$. Also, $\sum_{i \in T} z_i \geq 0$.

This holds for an arbitrarily small $\epsilon>0$, and the respective coefficient vector $(\mu_1, \ldots, \mu_m)$ is contained in a compact set $[0,d]^{m}$. Hence, we can pick a convergent subsequence as $\epsilon \to 0$ and conclude that the bound holds for some $(\mu_1, \ldots, \mu_m)$ with $\epsilon=0$.
\end{proof}

\subsection{Aggregation of $d \times d$ contributions}

Now we can show that $\sum_{|T|=d} S_T^4(\bx)$ attains a good value relative to the relaxation.

\begin{lemma}
\label{lem:4th-moment}
For a fractional solution $\bx$ satisfying $D(\bx) \geq e^{-d} \cR$,
$$ \sum_{|T|=d} S_T^4(\bx) \geq e^{-7d} \cR.$$
\end{lemma}

\begin{proof}
We have $$D(\bx) = \sum_{|T|=d} D_T(\bx) \leq \sum_{|T|=d} M_T^{3/4} S_T(\bx).$$
We apply H\"{o}lder's inequality, $\sum_T m_T s_T \leq (\sum m_T^\alpha)^{1/\alpha} (\sum s_T^\beta)^{1/\beta}$, with $m_T = M_T^{3/4}$, $s_T = S_T(\bx)$, $\alpha=4/3$, $\beta = 4$:
$$ D(\bx) \leq \left( \sum_{|T|=d} M_T \right)^{3/4} \left( \sum_{|T|=d} S_T^4(\bx) \right)^{1/4}.$$
Using Lemma~\ref{lem:KKT}, we get $$D(\bx) \leq \cR^{3/4}  \left( \sum_{|T|=d} \prod_{i \in T} \mu_i \right)^{3/4} \left( \sum_{|T|=d} S_T^4(\bx) \right)^{1/4}.$$
We estimate the first factor using the condition $\sum_{i=1}^{m} \mu_i = d$:
$\sum_{|T|=d} \prod_{i \in T} \mu_i \leq \frac{1}{d!} (\sum_{i=1}^{m} \mu_i)^d = \frac{d^d}{d!} \leq e^d$.
Therefore, $$ e^{-d} \cR \leq D(\bx) \leq e^{3d/4} \cR^{3/4} \left( \sum_{|T|=d} S_T^4(\bx) \right)^{1/4} $$
which proves the lemma.
\end{proof}

Now we have all the components in place to prove the main guarantee for our solution. This still requires some work, because our tools provide a bound on the logarithmic expectation of $\det (\sum_{i \in T} \bw_i \bw_i^T)$. To finish the analysis, we need to relate this to the logarithmic expectation of the global determinant, $\det (\sum_{i=1}^{m} \bw_i \bw_i^T)$.

\begin{theorem}
\label{thm:stable-chaos-m}
Consider a fractional solution $\bx$ satisfying $D(\bx) \geq e^{-d} \cR$.
Let $\bw_1,\ldots,\bw_m$ be the random vectors constructed by the algorithm. Then with constant probability,
$$ \det \left( \sum_{i=1}^{m} \bw_i \bw_i^T \right) \geq e^{-9d} \cR.$$
\end{theorem}

\begin{proof}
Using the Cauchy-Binet expansion and once again the Gibbs inequality (Proposition~\ref{prop:Gibbs}), with $p_T$ uniform, $\gamma_T = \det (\sum_{i \in T} \bw_i \bw_i^T)$ and $q_T$ to be determined:
$$ \log \det \left( \sum_{i=1}^{m} \bw_i \bw_i^T \right) =  \log \sum_{\substack{T \subseteq [m] \\ |T|=d}} \det \left( \sum_{i \in T} \bw_i \bw_i^T \right) \geq 
\sum_{\substack{T \subseteq [m] \\ |T|=d}} q_T \log \det \left( \sum_{i \in T} \bw_i \bw_i^T \right) 
 - \sum_{\substack{T \subseteq [m] \\ |T|=d}} q_T \log q_T. $$
Our goal is to apply Lemma~\ref{lem:stable-chaos} to analyze the contribution from each fixed $T$. Note that since $|T|=d$, the approximation factors for each $T$ will depend on $d$ rather than $m$.
From the algorithm, we have $\bw_i = \frac{1}{\sum_{j \in J_i} x_{ij}^2 |Z_{ij}|} \sum_{j \in J_i} x_{ij}^2 Z_{ij} \bv_{ij}$. 
Hence
$$ \det \left( \sum_{i \in T} \bw_i \bw_i^T \right) = \left( \frac{Q_T(\bx)}{ \prod_{i \in T} \sum_{j \in J_i} x_{ij}^2 |Z_{ij}|} \right)^2 $$
where
$$Q_T(\bx) = \sum_{\substack{j_i \in J_i \\ \forall i \in T}} \left(\prod_{i \in T} x_{i,j_i}^2 Z_{i,j_i} \right) a_{\bj,T}, $$
and $a_{\bj,T} = \det [\bv_{i j_i}: i \in T]$ is the determinant of a matrix whose columns are the vectors $\bv_{ij_i}$ indexed by $i \in T$ (let's say in increasing order; we take a square of the determinant anyway). 
We have
$$ \E\left[ \log \det \sum_{i \in T} \bw_i \bw_i^T\right] = 2\E[ \log |Q_T(\bx)| ] - 2\sum_{i \in T} \E\left[\log \sum_{j \in J_i} x_{ij}^2 |Z_{ij}| \right] \geq 2\E[\log |Q_T(\bx)|] - 2d (\kappa + \log 2).  $$
We also define, 
$$S_T(\bx) = \sum_{\substack{j_i \in J_i \\ \forall i \in T}} \left(\prod_{i \in T} x_{i,j_i} \right) |a_{\bj,T}|^{1/2} $$
Note that all the expressions here depend only on $x_{ij}$ for  $i \in T$. By the Lemma~\ref{lem:stable-chaos}, applied to the process restricted to $T$,
$$ \E[\log |Q_T(\bx)|] \geq 2 \log S_T(\bx) + d \kappa.$$
Hence, 
$$ \E\left[ \log \det \sum_{i \in T} \bw_i \bw_i^T\right] \geq 4 \log S_T(\bx) - 2d \log 2.$$
Now we go back to the global quantity $\log \det (\sum_{i=1}^{m} \bw_i \bw_i^T)$, and take the expectation:
$$ \E\left[ \log \det \left( \sum_{i=1}^{m} \bw_i \bw_i^T \right) \right] \geq 
\sum_{\substack{T \subseteq [m] \\ |T|=d}} q_T \E\left[ \log \det \left( \sum_{i \in T} \bw_i \bw_i^T \right) \right]  - \sum_{\substack{T \subseteq [m] \\ |T|=d}} q_T \log q_T
 $$
$$\geq \sum_{\substack{T \subseteq [m] \\ |T|=d}} q_T \log S_T^4(\bx) - 2d \log 2  - \sum_{\substack{T \subseteq [m] \\ |T|=d}} q_T \log q_T.$$
We can choose the distribution $q_T$ to maximize this expression; by the Gibbs variational principle, the optimal choice is $q_T = \frac{S_T^4(\bx)}{\sum_{T'} S_{T'}^4(\bx)}$. This gives
$$ \E\left[ \log \det \left( \sum_{i=1}^{m} \bw_i \bw_i^T \right) \right] \geq 
 \sum_{\substack{T \subseteq [m] \\ |T|=d}} \frac{S_T^4(\bx)}{\sum_{T'} S_{T'}^4(\bx)} \log S_T^4(\bx) - 2d \log 2  - \sum_{\substack{T \subseteq [m] \\ |T|=d}} \frac{S_T^4(\bx)}{\sum_{T'} S_{T'}^4(\bx)} \log \frac{S_T^4(\bx)}{\sum_{T'} S_{T'}^4(\bx)} $$
 $$ = \log \sum_{\substack{T \subseteq [m] \\ |T|=d}} S_T^4(\bx) - 2d \log 2
 \geq  \log \cR - 7d - 2d \log 2 $$
 by Lemma~\ref{lem:4th-moment}. Also, by Lemma~\ref{lem:rounding-m}, the value of $\det \left( \sum_{i=1}^{m} \bw_i \bw_i^T \right)$ cannot be larger than $\cR$.
Since $7 + 2 \log 2 < 9$, by Markov's inequality, $\log \det \left( \sum_{i=1}^{m} \bw_i \bw_i^T \right) > \log \cR - 9d$ with constant probability.
\end{proof}

Hence, we obtain vectors $\bw_1,\ldots,\bw_m$ with a good value $\det (\sum_{i=1}^{m} \bw_i \bw_i^T)$ with constant probability. 
Together, Theorem~\ref{thm:stable-chaos-m} and Lemma~\ref{lem:rounding-m} imply the simple partition matroid case of Theorem~\ref{thm:main}.

\section{Algorithm and analysis for $m\le d$}

In the case where we choose $1\le m\le d$ vectors (one from each
part), the objective is replaced by the $m$-dimensional unsigned volume of
the chosen vectors. For $U=[\bu_1,\ldots,\bu_m]\in\R^{d\times m}$,
write
\[
 \Vol_m[\bu_1,\ldots,\bu_m]:=\sqrt{\det(U^TU)}.
\]
For $\bj=(j_1,\ldots,j_m)\in J_1\times\cdots\times J_m$, set
\[
 a_{\bj}:=\Vol_m[\bv_{1,j_1},\ldots,\bv_{m,j_m}],
 \qquad
 \OPT:=\max_{\bj}a_{\bj}^2.
\]
Denote
$x_{\bj}=\prod_i x_{i,j_i}$ and $Z_{\bj}=\prod_i Z_{i,j_i}$.
We assume $\OPT>0$, otherwise the problem is trivial.
We run the same algorithm as in the $m=d$ case, replacing absolute
determinants by $m$-dimensional volumes. Given probability vectors
$\bx$, we define
\begin{equation}\label{eq:lower-rank-columns}
 \beta_i:=\sum_{j\in J_i}x_{ij}^2|Z_{ij}|,
 \qquad
 \bw_i:=\frac{\sum_{j\in J_i}x_{ij}^2Z_{ij}\bv_{ij}}{\beta_i},
\end{equation}
where we sample independent copies $Z_{ij}\sim\cD_{1/2}$. Our starting point is again the fractional relaxation of Nikolov and Singh. For a positive semidefinite $A\in\R^{d\times d}$, let
$\operatorname{sym}_m(A)$ be the $m$-th elementary symmetric polynomial of its eigenvalues. Define
\begin{equation}\label{eq:lower-saddle-point}
 \cR:=\max_{\bx}\inf_{\bz:\,\sum_{i=1}^m z_i\ge0}
\operatorname{sym}_m\left(\sum_{i=1}^m\sum_{j\in J_i}
 e^{z_i}x_{ij}\bv_{ij}\bv_{ij}^T\right)\ge \OPT,
\end{equation}
where $x_{ij}\ge0$ and $\sum_{j\in J_i}x_{ij}=1$ for every $i\in[m]$.
This is the simple partition matroid specialization of \cite[Section~3.2, Lemma~3.3]{NS16}. When $m=d$, $\operatorname{sym}_d=\det$ and the earlier relaxation is recovered. \cite[Lemma~3.6]{NS16} gives an efficient algorithm to compute a fractional solution $\bx$ such that
\begin{equation}\label{eq:lower-NS-input}
 D(\bx):=\sum_{\bj}x_{\bj}a_{\bj}^2
 \ge e^{-m}\cR.
\end{equation}
We now present the algorithm for determinant maximization in the case of $m\le d$. It is essentially identical to the $m= d$ case.

\begin{algorithm}[H]
\caption{Randomized rounding for a simple partition matroid, $m\le d$}
\label{alg:lower-rank-rounding}
\begin{algorithmic}[1]
\Require Vectors $\bv_{ij}\in\R^d$, near-optimal fractional solution $(\bx,\bz)$, and
$\cR>0$ satisfying \eqref{eq:lower-NS-input}.
\Ensure One index $j(i)\in J_i$ chosen for each $1\le i\le m$.
\Repeat
    \State Draw independent random variables
    $Z_{ij}\sim\cD_{1/2}$ for $1\le i\le m$, $j\in J_i$.
    \State $y_{ij}\gets
    x_{ij}^2Z_{ij}/\sum_{k\in J_i}x_{ik}^2|Z_{ik}|$.
    \State $\bw_i\gets\sum_{j\in J_i}y_{ij}\bv_{ij}$.
\Until{$\Vol_m[\bw_1,\ldots,\bw_m]
       \ge e^{-3m}\sqrt{\cR}$}
\For{$i=1,\ldots,m$}
    \State Choose $j(i)\in J_i$ to maximize
    $
     \Vol_m[
       \bv_{1,j(1)},\ldots,\bv_{i,j(i)},
       \bw_{i+1},\ldots,\bw_m]$.
\EndFor
\State \Return $(j(1),\ldots,j(m))$.
\end{algorithmic}
\end{algorithm}

Again, we first observe the iterative rounding consequence.

\begin{lemma}
\label{lem:lower-rank-output}
If Algorithm~\ref{alg:lower-rank-rounding} terminates, it returns
a solution such that
\[
 \Vol_m[\bv_{1,j(1)},\ldots,\bv_{m,j(m)}]^2
 \ge e^{-6m}\OPT.
\]
\end{lemma}
\begin{proof}
Fix some matrix $B$ with $m-1$ columns. Observe the elementary identity
\[
 \Vol_m[B,\bu]
 =
 \sqrt{\det(B^TB)}\,
 \|\Pi_{\operatorname{span}(B)^\perp}\bu\|.
\]
Since $\bw_i=\sum_jy_{ij}\bv_{ij}$ and $\sum_j|y_{ij}|=1$,
the triangle inequality gives
\[
 \Vol_m[B,\bw_i]
 \le
 \sum_j|y_{ij}|\Vol_m[B,\bv_{ij}]
 \le
 \max_{j\in J_i}\Vol_m[B,\bv_{ij}].
\]
Therefore each iterative replacement cannot decrease volume.
On termination, we see that
\[
 \Vol_m[\bv_{1,j(1)},\ldots,\bv_{m,j(m)}]
 \ge
 \Vol_m[\bw_1,\ldots,\bw_m]
 \ge
 e^{-3m}\sqrt{\cR}\ge e^{-3m}\sqrt{\OPT}.
\]
\end{proof}

\subsection{Analysis of the sampling procedure}
Our remaining goal is to prove that the algorithm will succeed in Steps 2-5 with good probability. This is captured by the following.

\begin{theorem}
\label{thm:lower-rank-sampling}
Consider an optimal fractional solution $\bx$ of
relaxation \eqref{eq:lower-saddle-point} whose value is $\cR>0$.
Let $\bw_1,\ldots,\bw_m$ be the random vectors constructed in Algorithm~\ref{alg:lower-rank-rounding}. Then,
with constant probability,
\[
 \Vol_m[\bw_1,\ldots,\bw_m]
 \ge e^{-3m}\sqrt{\cR}.
\]
\end{theorem}

As in the $m=d$ case, we introduce the deterministic quantity $S(\bx)$ below, and observe from \eqref{eq:lower-NS-input} that $a_{\bj}^2\le\cR$, so
\begin{equation}\label{eq:lower-S-comparison}
 \cR^{1/4}\ge S(\bx):=\sum_{\bj}x_{\bj}a_{\bj}^{1/2}
 \ge\frac{D(\bx)}{\cR^{3/4}}
 \ge e^{-m}\cR^{1/4}.
\end{equation}
It therefore suffices to prove the following analogue of the
logarithmic stable chaos inequality (Lemma~\ref{lem:stable-chaos}) but in the $m\le d$ case with volumes. The proof is deferred to the next section.

\begin{lemma}[Logarithmic stable volume]
\label{lem:stable-vol}
For any probability vectors $\bx$ with $S(\bx)>0$, the columns
$\bw_1,\ldots,\bw_m$ in \eqref{eq:lower-rank-columns} satisfy
\begin{equation}\label{eq:lower-normalized-chaos}
 \E_Z[\log\Vol_m[\bw_1,\ldots,\bw_m]]  \ge 2\log S(\bx)-m\log2.
\end{equation}
\end{lemma}

Assuming Lemma~\ref{lem:stable-vol}, \eqref{eq:lower-S-comparison}
implies
\[
 \E_Z[\log\Vol_m[\bw_1,\ldots,\bw_m]]
 \ge\tfrac12\log\cR-(2+\log2)m,
\]
and Lemma~\ref{lem:lower-rank-output} implies $\Vol_m[\bw_1,\ldots,\bw_m]\le\sqrt{\OPT}\le\sqrt{\cR}$. 
Thus, the same Markov argument as in Theorem~\ref{thm:stable-chaos} gives
\[
 \Pp\left(
   \Vol_m[\bw_1,\ldots,\bw_m]\ge e^{-3m}\sqrt{\cR}
 \right)
 \ge1-\frac{2+\log2}{3}
 =\frac{1-\log2}{3}>0.
\]
Thus an expected constant number of trials suffices. On termination, the output attains a squared volume of at least $e^{-6m}\cdot \OPT$ by Lemma~\ref{lem:lower-rank-output}. This proves Theorem~\ref{thm:lower-rank-sampling}, and therefore the
simple-partition case of Theorem~\ref{thm:main2} after taking square roots, subject only to
Lemma~\ref{lem:stable-vol}.

\subsection{The Gaussian transformation}

We prove Lemma~\ref{lem:stable-vol} in this section by reducing it to Lemma~\ref{lem:stable-chaos} via a Gaussian transformation. 
We emphasize that transformation is only a proof device, not a part of the algorithm.

The reduction is based on the following observation.
We say that a matrix is \emph{standard Gaussian} if all its entries are i.i.d. standard Gaussian. Let $\Vol_m(U)$ be the volume of the $m$ columns of $U$.

\begin{lemma} \label{lem:lower-gaussian} Let $T\in\R^{m\times d}$ and $H\in\R^{m\times m}$ be standard Gaussian. Then $\E[|\log|\det H||]<\infty$. With $c_m:=\E[\log|\det H|]$, every fixed full-column-rank $U\in\R^{d\times m}$ satisfies \begin{equation}\label{eq:lower-gaussian-law} |\det(TU)|\stackrel{\mathrm d}{=} \Vol_m(U)|\det H|, \qquad \E_T[\log|\det(TU)|]=\log\Vol_m(U)+c_m. \end{equation} \end{lemma}

\begin{proof}
The finiteness of $\E[|\log|\det H||]$ (and therefore $c_m$) follows from standard
log-determinant integrability of Gaussian matrices. This is given in, e.g., \cite[Eq.~(3.3)]{DE13} and \cite[Eq.~(40)]{CLZ}.

To prove \eqref{eq:lower-gaussian-law}, we apply the QR-factorization of rectangular matrices to obtain
$U=QR$, where $Q^TQ=I_m$ and $R$ is invertible.
Then $TQ$ is a standard $m\times m$ Gaussian matrix and
\[
  |\det R|=\sqrt{\det(U^TU)}=\Vol_m(U)
  \implies
  |\det(TU)|\stackrel{\mathrm d}{=}
  \Vol_m(U)|\det H|.
\]
Taking logarithms and expectations now yields
\[
  \E_T[\log|\det(TU)|]
  =\log\Vol_m(U)+c_m.\qedhere
\]
\end{proof}



This lemma relates the $m$-dimensional volume in $\R^d$ to the volume in $\R^m$ after a Gaussian transformation. The latter is then a simple determinant, so our $m=d$ case results apply. 

\begin{proof}[Proof of Lemma~\ref{lem:stable-vol}]
Sample standard Gaussian $T\in\R^{m\times d}$, independently of
the $Z_{ij}$. Define
\[
 a_{\bj}^{(T)}:= \det[T\bv_{1,j_1},\ldots,T\bv_{m,j_m}]\quad\text{and}\quad 
 S^{(T)}(\bx):=\sum_{\bj}x_{\bj}|a_{\bj}^{(T)}|^{1/2}.
\]
The transformed problem has $m$ parts
in dimension $m$. By multilinearity,
\[
 \det[T\bw_1,\ldots,T\bw_m]
 =
 \frac{\sum_{\bj}a_{\bj}^{(T)}x_{\bj}^2Z_{\bj}}
      {\prod_{i=1}^m\beta_i}.
\]
For any $T$, Lemma~\ref{lem:1/2-stable} and Lemma~\ref{lem:stable-chaos} with scalar coefficients $\bj\mapsto a_{\bj}^{(T)}$ implies
\begin{equation}
 \E_Z[\log|\det[T\bw_1,\ldots,T\bw_m]|]
 \ge
 2\log S^{(T)}(\bx)+m\kappa-m(\kappa+\log2)=2\log S^{(T)}(\bx)-m\log2.
 \label{eq:lower-projected-balanced}
\end{equation}

We now relate the transformed quantities with $T$ to those without $T$.
For each $\bj$ with $a_{\bj}>0$, Lemma~\ref{lem:lower-gaussian}
immediately gives $ \E_T[\log|a_{\bj}^{(T)}|]=\log a_{\bj}+c_m$. 
To obtain the corresponding comparison for $S^{(T)}(\bx)$, define a probability distribution over indices $\bj$ via
\[
 q_{\bj}:=\frac{x_{\bj}a_{\bj}^{1/2}}{S(\bx)}.
\]
Restricting the following sums to
$x_{\bj}a_{\bj}>0$, Proposition~\ref{prop:Gibbs}
gives
\[
 2\log S^{(T)}(\bx)
 =2\log S(\bx)
   +2\log\left(
     \sum_{\bj}q_{\bj}
     \left|\frac{a_{\bj}^{(T)}}{a_{\bj}}\right|^{1/2}
   \right)\ge
 2\log S(\bx)
 +\sum_{\bj}q_{\bj}\log\frac{|a_{\bj}^{(T)}|}{a_{\bj}}.
\]
Taking expectation over $T$, each logarithm in the last sum
has expectation $c_m$. Hence
\begin{equation}\label{eq:lower-projected-S}
 2\E_T[\log S^{(T)}(\bx)] \ge 2\log S(\bx)+c_m.
\end{equation}
Finally, apply Lemma~\ref{lem:lower-gaussian} conditionally on
$\bw_1,\ldots,\bw_m$, and combine
\eqref{eq:lower-projected-balanced} and \eqref{eq:lower-projected-S}.
We remark that whenever $S(\bx) \neq 0$, 
$\rank [\bw_1,\ldots, \bw_m] = m$ almost surely, because there is 
a non-singular $m \times m$ minor and the continuously random choice
of $\bw_1,\ldots,\bw_m$ is almost surely linearly independent.
\begin{align*}
 \E_Z[\log\Vol_m[\bw_1,\ldots,\bw_m]]+c_m
 &=\E_T\E_Z[\log|\det[T\bw_1,\ldots,T\bw_m]|]\\
 &\ge2\E_T[\log S^{(T)}(\bx)] - m\log2\\
 &\ge2\log S(\bx) +c_m - m\log2.
\end{align*}
We were able to exchange $\E_T$ and $\E_Z$ because $\log |\det [T\bw_1,\ldots,T\bw_m]|$
is integrable by similar arguments as before.
The lemma follows by canceling $c_m$ on both sides.
\end{proof}

\section{Reduction from general partition matroids to simple partition matroids}

Finally, we come to the setting of more general partition matroids. Here, we have disjoint sets $J_1,\ldots,J_r$, capacity parameters $k_i \geq 1$ and vectors $\bv_{ij} \in \R^d$ for $1 \leq i \leq r$, $j \in J_i$. The goal is to choose $K_i \subseteq J_i$ such that $|K_i| \leq k_i$, in order to maximize the respective determinant or volume. We denote by $m = \sum_{i=1}^{r} k_i$ the rank of the respective partition matroid. We show how to reduce this problem to a simple partition matroid constraint. First let us consider a case where the reduction is very simple.

\subsection{The $m \leq d$ case}

Here, we aim to maximize the $m$-dimensional volume formed by the selected vectors:
$$\Vol_m [\bv_{ij}: i \in [r], j \in K_i].$$
In this case, we have a simple lemma.

\begin{lemma}
If there is an $\alpha$-approximation determinant maximization subject to a simple partition matroid in the $m \leq d$ in the case, then there is also an $\alpha$-approximation in the same regime subject to a general partition matroid.
\end{lemma}

\begin{proof}
Consider a partition matroid with parts $J_i$ and capacity parameters $k_i$.
For each set $J_i$, create $k_i$ copies $J_i^{(1)}, \ldots, J_i^{(k_i)}$, and for each $\bv_{ij}, j \in J_i$, create respective copies $\bv_{i,s,j} = \bv_{ij}$ where $1 \leq s \leq k_i$. We consider a simple partition matroid on these new index sets: we want to select one index from each $J_i^{(s)}$; $m = \sum_{i=1}^{r} k_i$ indices in total. 
Any solution selecting two copies of the same index $j$ from $J_i^{(s)}$ and $J_i^{(s')}$ has value $0$, by the definition of $m$-dimensional volume. Hence, solutions of nonzero value correspond to choosing $k_i$ distinct indices from $J_i$, as in the original partition matroid. The optimum is the same, and a solution found by an algorithm for the simple partition matroid can be easily converted to a solution of the same value for the original problem.
\end{proof}

\subsection{The $m > d$ case}

The reduction in this case is a bit more involved. Here we optimize the following objective:
$$ \det \left( \sum_{i=1}^{r} \sum_{j \in K_i} \bv_{ij} \bv_{ij}^T \right).$$
We prove the following general reduction.

\begin{theorem}
If there is an $\alpha$-approximation algorithm for determinant maximization subject to a simple partition matroid constraint, then there is an $\alpha e^{O(d)}$-approximation algorithm for determinant maximization subject to a general partition matroid.
\end{theorem}

\begin{proof}
We design a randomized reduction. Given sets $J_i$ with capacity parameters $k_i \geq 1$,
we create again copies $J_i^{(s)} = \{ (i,s,j): j \in J_i\}$ for $1 \leq s \leq k_i$ and consider a simple partition matroid on the sets $J_i^{(s)}$. The duplication of vectors will involve some scaling. For each vector $\bv_{ij}, j \in J_i$, we create vectors 
$$ \bv_{i,s,j} = \frac{1}{\pi_{ij}(s)} \bv_{i,j} $$
where $1 \leq s \leq k_i$ and $\pi_{ij}: [k_i] \to [k_i]$ is an independently random permutation for each pair $(i,j)$. 
In other words, we scale the vector copies by reciprocals of randomly shuffled integers $1,2,\ldots,k_i$. We have two claims.

\paragraph{Claim 1.}
For any feasible solution $I'$ in the new instance, there is a solution $I$ in the original instance such that 
$$ \det \left( \sum_{(i,j) \in I} \bv_{ij} \bv_{ij}^T \right) \geq \left( \frac{6}{\pi^2} \right)^d \det \left( \sum_{(i,s,j) \in I'} \bv_{i,s,j} \bv_{i,s,j}^T \right). $$

\paragraph{Proof:}
Observe that copies of the same vector $\bv_{ij}$ can contribute at most
$$ \sum_{s=1}^{k_i} \bv_{i,s,j} \bv_{i,s,j}^T = \sum_{s=1}^{k_i} \frac{1}{s^2} \bv_{ij} \bv_{ij}^T
\leq \frac{\pi^2}{6} \bv_{ij} \bv_{ij}^T $$
(with the inequality denoting PSD ordering). 
By Cauchy-Binet, this factor be amplified by a power of $d$ in the global determinant; 
if we select $(i,j)$ in $I$ whenever at least one copy $(i,s,j)$ is selected in $I'$, then
$$ \det \left( \sum_{(i,j) \in I} \bv_{ij} \bv_{ij}^T \right)
 \geq \left( \frac{6}{\pi^2} \right)^d \det \left( \sum_{(i,s,j) \in I'} \bv_{i,s,j} \bv_{i,s,j}^T \right). $$

\paragraph{Claim 2.}
For any solution $I$ of the original instance, there is a solution $I'$ of the new instance (depending on the random permutations $\pi_{ij}$) such that
$$ \E_\pi\left[ \log \det \left( \sum_{(i,s,j) \in I'} \bv_{i,s,j} \bv_{i,s,j}^T \right) \right] \geq \log \det \left( \sum_{(i,j) \in I} \bv_{ij} \bv_{ij}^T \right) -  2d.$$

\paragraph{Proof:}
Let $I$ be a solution of the original instance (w.l.o.g. a basis of the respective matroid). Given the random permutations $\pi_{ij}$, we define a solution $I'$ of the new instance as follows: Processing the pairs $(i,j) \in I$ in a random order, we select the triple $(i,s,j)$ with the minimum value of $\pi_{ij}(s)$ such that we haven't selected any $(i,s,\cdot)$ triple yet. Since there are at most $k_i$ pairs $(i,\cdot)$ in $I$, there is always some available value of $s \in [k_i]$; denote this value $s_{ij}$. This way, we construct an independent set in $\cI'$. 

What is the expected value of $\pi_{ij}(s_{ij})$? Assuming that this is the $\ell$-th element for a given $i$, the available slots $(i,s)$ form a uniformly random $(k_i-\ell+1)$-subset of the $k_i$ indices $1 \leq s \leq k_i$ (since there is no correlation between the permutation $\pi_{ij}$ and the slots occupied so far). The expected minimum among $k_i-\ell+1$ random elements in $[k_i]$ is $\E_\pi[\pi_{ij}(s_{ij})] = \frac{k_i+1}{k_i-\ell+2}$. Therefore, by Jensen's inequality,
$$ \E\pi[ \log \pi_{ij}(s_{ij})] \leq \log \frac{k_i+1}{k_i-\ell+2} \leq \log \frac{k_i}{k_i-\ell+1} $$
for the $\ell$-th arriving element in $J_i$. Given the random ordering in which we process the elements, $\ell$ is uniformly random in $[k_i]$. Therefore, the expected contribution for $(i,j)$ is
$$ \E[ \log \pi_{ij}(s_{ij})] \leq \frac{1}{k_i} \sum_{\ell=1}^{k_i} \log \frac{k_i}{k_i-\ell+1} = \frac{1}{k_i} \log \frac{k_i^{k_i}}{k_i!} \leq 1.$$
The assigned vector is $\bv_{i,s_{ij},j} = \frac{\bv_{i,j}}{\pi_{ij}(s_{ij})}$,
and each rank 1 matrix $\bv_{i,s_{ij},j} \bv_{i,s_{ij},j}^T$ is scaled by $\lambda_{ij} = \frac{1}{(\pi_{ij}(s_{ij}))^2}$.
By Cauchy-Binet, the expected value of the transformed solution $\sum_{(i,s,j) \in I'} \bv_{i,s,j} \bv_{i,s,j}^T$ can be decomposed into contributions over subsets of size $d$, and the scaling factor for each such subset $T$ satisfies
$ \E[\log \prod_{(i,j) \in T} \lambda_{ij}] = \sum_{(i,j) \in T} \E[\log \lambda_{ij}] \geq -2d.$
Hence, for every $d$-subset $T$ with $\det \left( \sum_{(i,j) \in T} \bv_{ij} \bv_{ij}^T \right)>0$, we have
$$ \E \left[ \log \frac{\det \left( \sum_{(i,j) \in T} \bv_{i,s_{ij},j} \bv_{i,s_{ij},j}^T \right)}{\det \left( \sum_{(i,j) \in T} \bv_{ij} \bv_{ij}^T \right)} \right] \geq -  2d.$$
Adding up over all $T, |T|=d$, with probability coefficients $p_T = \frac{\det \left( \sum_{(i,j) \in T} \bv_{ij} \bv_{ij}^T \right)}{ \det \left( \sum_{(i,j) \in I} \bv_{ij} \bv_{ij}^T \right)}$, and using
again Jensen's inequality,
$$ \E\left[ \log \frac{\det \left( \sum_{(i,s,j) \in I'} \bv_{i,s,j} \bv_{i,s,j}^T \right)}{\det \left( \sum_{(i,j) \in I} \bv_{ij} \bv_{ij}^T \right)} \right]
= \E\left[ \log \frac{ \sum_{|T|=d} \det \left( \sum_{(i,j) \in T} \lambda_{ij}  \bv_{i,j} \bv_{i,j}^T \right)}{\det \left( \sum_{(i,j) \in I} \bv_{ij} \bv_{ij}^T \right)} \right] $$
$$ = \E\left[ \log \sum_{|T|=d} p_T \prod_{(i,j) \in T} \lambda_{ij} \right]
\geq \sum_{|T|=d} p_T \E\left[ \log \prod_{(i,j) \in T} \lambda_{ij} \right] \geq -2d.$$
This proves Claim 2.

The ratio of the new objective to the old one is at most $1$, because each rank one matrix is scaled by a factor $\leq 1$. Hence the logarithm inside the expectation is at most $0$.
By Markov's inequality, we find a transformation such that the right-hand side is at least $-3d$ with constant probability. Hence, given an optimal solution for the original instance $\OPT$, there is a solution for the new instance of value at least $e^{-3d} \OPT$. Our $\alpha$-approximation algorithm for simple partition matroids can extract a solution of value at least $\frac{1}{\alpha} e^{-3d} \OPT$, and finally we convert it to a solution of the original instance of value at least $\frac{1}{\alpha} e^{-3d} (6/\pi^2)^d \OPT$.

\end{proof}

\subsection*{Funding}
YS is funded by the NSF Graduate Research Fellowship Program and the Stanford Graduate Fellowship.
\bibliographystyle{plain}
\bibliography{ref.bib}

\appendix

\section{Properties of $1/2$-stable random variables}
\label{sec:stable-facts}

Here we prove Lemma~\ref{lem:1/2-stable}.

\begin{proof}
We have $Z = \frac{1}{2 X_1^2} - \frac{1}{2 X_2^2} = \frac{(X_2+X_1)(X_2 - X_1)}{2 X_1^2 X_2^2}$ where $X_1, X_2$ are independent standard Gaussians.
From here,
\[
 \log|Z|
 =
 \log|X_2+X_1|+\log|X_2-X_1|
 -\log2-2\log|X_1|-2\log|X_2|.
\]
Recall that nondegenerate Gaussian variables have absolutely integrable logarithms,
and so $\E[|\log|Z||]<\infty$. Moreover, $X_2\pm X_1$ each has
the distribution of $\sqrt2X_1$, so
\[
 \kappa:=\E[\log|Z|]
 =-2\E[\log|X_1|]
\simeq1.27.
\]
This proves the second bullet point. Now, let $P=1/(2X^2)$ for a standard Gaussian $X$. Its Laplace transform is
\[
 \E[e^{-sP}]
 =
 \sqrt{\frac2\pi}
 \int_0^\infty
 \exp\left(-\frac{t^2+s/t^2}{2}\right)\,dt
 =
 e^{-\sqrt s},
\]
via standard Gaussian integral computations.
By independence and uniqueness of Laplace transforms,
\[
 \sum_i b_iP_i
 \stackrel{\mathrm d}{=}
 \left(\sum_i\sqrt{b_i}\right)^2P,
\]
for any $b_i\ge 0$ and independent copies $P_i$  of $P$.
Sample $Z_i=P_i-P_i'$ using mutually independent copies of $P$.
Symmetry allows us to absorb the signs of the coefficients, giving
\[
 \sum_i c_iZ_i
 \stackrel{\mathrm d}{=}
 \sum_i|c_i|P_i-\sum_i|c_i|P_i'
 \stackrel{\mathrm d}{=}
 c(P-P')
 \stackrel{\mathrm d}{=}cZ.
\]
This proves the first bullet point.
For the third bullet point, pointwise,
\[
 \frac1c\left|\sum_i c_iZ_i\right|
 \le Z'
 \le
 \frac1c\sum_i|c_i|(P_i+P_i').
\]
The left-hand side has distribution $|Z|$, while 
the right-hand side has distribution $ P+P'\stackrel{\mathrm d}{=}4P$.
Both sides have absolutely integrable logarithms, since
$\E[|\log|Z||]<\infty$ and $\log P=-\log2-2\log|X|$.
Therefore, $\E[|\log Z'|]<\infty$. Finally, we bound
\[
 \E[\log Z']
 \le
 \E[\log(4P)]
 =
 \log4-\log2-2\E[\log|X|]
 =
 \kappa+\log2.\qedhere
\]
\end{proof}

\section{Integrality gap for $m > d$}
\label{sec:integrality-gap}

Here we justify the claim that the integrality gap of the saddle-point relaxation for simple partition matroids is at most $e^d$, even for $m > d$. This follows from known work, but it is not explicitly stated to our knowledge (a bound of $e^{O(d)}$ is stated in \cite{MNST20}).

Let us write the relaxation as follows:
\begin{eqnarray*}
\cR  = & \max R(\bx): \\
\forall i \in [m]; & \sum_{j \in J_i} x_{ij} = 1, \\
\forall i \in[m], j \in J_i; & x_{ij} \geq 0, \\
\end{eqnarray*}
where
\begin{eqnarray*}
R(\bx) = \inf_\bz & \det \left (\sum_{i=1}^{m} \sum_{j \in J_i} e^{z_i} x_{ij} \bv_{ij} \bv_{ij}^T \right) \\
\forall I \subseteq [m], |I|=d; & \sum_{i \in I} z_i \geq 0.
\end{eqnarray*}
We also define:
$$ D(\bx) = \sum_{j_1 \in J_1,\ldots, j_m \in J_m} \left( \prod_{i=1}^{m} x_{i,j_i} \right) \det \left( \sum_{i=1}^{m} \bv_{i,j_i} \bv_{i,j_i}^T \right).$$
The (discrete) optimum value is
\[
 \OPT=\max_{\substack{j_i\in J_i \\ \forall i \in [m]}}
 \det\left(\sum_{i=1}^m \bv_{i j_i}\bv_{i j_i}^{\mathsf T}\right).
\]

\paragraph{Why this is a relaxation.}
Consider a feasible solution, with $\bv_{ij_i}$ being the vector selected in part $i$.
Cauchy--Binet gives
\[
 \det\left(\sum_{i=1}^{m} e^{z_i} \bv_{ij_i} \bv_{ij_i}^T \right)
 =\sum_{\substack{T\subseteq[m]\\|T|=d}}
 e^{z(T)} (\det [\bv_{ij_i}]_{i\in T})^2.
\]
Every feasible $\bz$ satisfies $e^{z(T)} \ge1$, and $\bz=\bzero$ is feasible.
Thus $R(\bx)=\det \left(\sum_{i=1}^{m} \bv_{ij_i} \bv_{ij_i}^T \right)$ for every integral solution $\bx$, proving $\cR \ge \OPT$.

The integrality gap result is as follows.

\begin{theorem}\label{thm:integrality-gap}
Consider a simple partition matroid constraint with $m$ parts, vectors in $\R^d$, and let $k=m-d>0$. Define
\[
 \gamma_{m,d}=\frac{m!}{m^m}\frac{k^k}{k!},
\]
with $0^0$ interpreted as $1$. Assume that $D(\bx) > 0$. Then
\begin{equation}\label{eq:rounding-guarantee}
D(\bx) \ge \gamma_{m,d} R(\bx)\ge e^{-d} R(\bx).
\end{equation}
In particular, for a near-optimal fractional solution $(\bx,\bz)$ such that $R(\bx) \geq (1-\epsilon) \cR$, for sufficiently small $\epsilon>0$,  we have  $D(\bx) \geq e^{-d} \cR$.
\end{theorem}

\paragraph{Stable polynomial construction.}
Fix $\bx = (x_{ij})_{1 \leq i \leq m, j \in J_i}$ and abbreviate $A_i = \sum_{j \in J_i} x_{ij} \bv_{ij} \bv_{ij}^T$.
Define
\begin{equation}\label{eq:px}
 p(t_1,\ldots,t_m) = \det\left (\sum_{i=1}^m t_i A_i\right) 
\end{equation}
This is a homogeneous polynomial of degree $d$ with nonnegative coefficients. 
Also, it is known that this is a {\em stable polynomial}, meaning that
it is nonzero whenever every variable has a strictly positive imaginary part. 
This follows from the fact that each $A_i$ is positive semidefinite, and $\sum_{i} A_i$ is positive definite for non-trivial instances (where $D(\bx) > 0$).
If $t_i \in \C$, $\operatorname{Im}t_i>0$ for all $i$, then $\operatorname{Im} \sum_i t_i A_i$ is also positive definite. Consequently, $\sum_i t_i A_i$ is nonsingular, proving stability.

Let $c_T$ denote the coefficient of $\bt^T = \prod_{i \in T} t_i$ in $p(\bt)$.
By Cauchy--Binet, 
\begin{equation}\label{eq:coefficients}
 c_T = \sum_{\substack{j_i \in J_i \\ \forall i\in T}}  \left(\prod_{i\in T}x_{i j_i}\right) \left( \det [\bv_{i j_i}]_{i\in T} \right)^2.
\end{equation}
I.e., $D(\bx) = \sum_{\substack{T\subseteq[m]\\|T|=d}} c_T$.
Also, given the substitution $t_i = e^{z_i}$, 
we have $R(\bx) = \inf_{\substack{\bt>0 \\ \forall |T|=d; \bt^T \geq 1}} p(\bt)$.
Note that $p(\bt)$ can contain monomials other than $c_T \bt^T$ because
 $A_i$ can have rank greater than one. The key is to compare $D(\bx)$ and $R(\bx)$,
which is accomplished by the {\em capacity inequality}.

\paragraph{The capacity inequality.}
The following result from \cite{AGSS17} is a generalized form of Gurvits' capacity inequality \cite{Gurvits08}.

\begin{proposition}
\label{prop:degree-d}
Let $p$ be a homogeneous real stable polynomial of degree $d$ in $m\ge d$
variables, with nonnegative coefficients. Let $c_T$ be the coefficient of the monomial $\bt^T$,
and define the capacity of the polynomial as
\[
 C(p)=\inf_{\substack{\bt>\bzero\\\bt^T\ge 1 \forall |T|=d}}p(\bt).
\]
If $\sum_{|T|=d}c_T>0$, then, with $k=m-d$ and $0^0=1$,
\begin{equation}\label{eq:degree-d}
 \sum_{|T|=d}c_T\ge\frac{m!}{m^m}\frac{k^k}{k!}\,C(p)
 \ge e^{-d}C(p).
\end{equation}
\end{proposition}

The conditions are satisfied in our case: the polynomial is homogeneous real stable, and 
$D(\bx) = \sum_{|T|=d} c_T > 0$ by assumption.
Hence $D(\bx) = \sum_{|T|=d} c_T \geq \frac{m!}{m^m} \frac{k^k}{k!} C(p) = \gamma_{m,d} R(\bx)$. 

For a near solution $\bx$, we have $R(\bx) \geq (1-\epsilon) \cR$. Also, we have
$$ \log \frac{m!}{(m-d)!} = \sum_{k=m-d+1}^{m} \log k >
 \int_{m-d}^{m} \log x \, dx = m (\log m - 1) - (m-d) (\log (m-d) - 1).$$
Therefore,
$$ \log \gamma_{m,d} = \log \frac{m!}{m^m} \frac{(m-d)^{m-d}}{(m-d)!} > -d. $$
For small enough $\epsilon > 0$, $D(\bx) \geq \gamma_{m,d} R(\bx) \geq e^{-d} \cR$, absorbing the $\epsilon$ in the gap between $\gamma_{m,d}$ and $e^{-d}$.
This completes the proof of Theorem~\ref{thm:integrality-gap}.

\end{document}